\documentclass[12pt]{birkjour}

\usepackage{amssymb}
\usepackage{epstopdf,comment}
\usepackage{  mathrsfs, enumerate}
\usepackage{amsmath}

\newtheorem{theorem}{Theorem}
\newtheorem{definition}[theorem]{Definition}

\newtheorem{corollary}[theorem]{Corollary}
\newtheorem{lemma}[theorem]{Lemma}
\newtheorem{assumption}{Assumption}
\newtheorem{remark}[theorem]{Remark}

\numberwithin{theorem}{section}
\numberwithin{assumption}{section}

\newcommand{\req}[1]{Eq.\,(\ref{#1})}

\begin{document}

\title[Langevin equations: Higher-order approximations]{Langevin equations in the small-mass limit:\\
Higher-order approximations}

\author[Birrell]{Jeremiah Birrell}
\address{Department of Mathematics and Statistics\\
University of Massachusetts Amherst\\
Amherst, MA 01003, USA}
\email{birrell@math.umass.edu}

\author[Wehr]{Jan Wehr}
\address{Department of Mathematics\\
Program in Applied Mathematics\\
University of Arizona\\
Tucson, AZ, 85721, USA}
\email{wehr@math.arizona.edu}



\date{\today}

\begin{abstract}
We study the small-mass (overdamped) limit of Langevin equations for a particle in a potential and/or magnetic field with matrix-valued and state-dependent drift and diffusion.  We utilize a bootstrapping argument to derive a hierarchy of approximate equations for the position degrees of freedom that  are able to achieve accuracy  of order $m^{\ell/2}$ over compact time intervals for any  $\ell\in\mathbb{Z}^+$. This generalizes prior derivations of the homogenized equation for the position degrees of freedom in the $m\to 0$ limit, which result in order $m^{1/2}$  approximations.   Our results cover bounded forces, for which we prove convergence in $L^p$ norms, and unbounded forces, in which case we prove convergence in probability.
\end{abstract}
\maketitle

\section{Introduction}
 Langevin equations provide models of a diffusing particle;  a simple example, illustrating several typical ingredients, is the system of stochastic differential equations (SDE) 
\begin{align}\label{model_sys}
dq^m_t=v^m_t dt,\hspace{2mm} m dv^m_t=-\gamma v^m_t dt+\sigma dW_t,
\end{align}
 where $m$ is the mass of the particle (here and in the sequel we use a superscript to denote the $m$ dependence), $\gamma$ and $\sigma$ are the dissipation (or: drag) and diffusion coefficients respectively and $W_t$ is a Wiener process.  Pioneering work, including investigation of the small-mass limit, was done by Smoluchowski \cite{smoluchowski1916drei} and Kramers \cite{KRAMERS1940284}. A detailed discussion of the early literature can be found in \cite{Nelson1967}.

Works studying the small-mass limit of (various generalizations of) \req{model_sys} have rigorously established convergence  of the position degrees of freedom, $q^m_t$,  as $m\to 0$ to the solution, $q_t$, of a limiting SDE \cite{Sancho1982,volpe2010influence,Hottovy2014,herzog2015small,particle_manifold_paper,BirrellHomogenization}. Such problems fall under the umbrella of homogenization (see, for example, the recent sources \cite{fouque2007wave,pavliotis2008multiscale}) and so we refer to the SDE for $q_t$ as the  homogenized equation.  Moving beyond the homogenized equation, results have been proven regarding the small-mass limit in the sense of rough-paths \cite{FritzRoughPath}, the limit  of the joint distribution of position and (scaled) velocity \cite{BirrellPhaseSpace}, the limit of the invariant measures \cite{hu2017}, and the limit of the entropy production \cite{Birrell2018}.

In this paper we  build on the above small-mass limit results, specifically those in  \cite{BirrellHomogenization}, where it was proven that $q_t$ approximates $q_t^m$ with $O(m^{1/2})$ error over compact time intervals (see also the summary in Section \ref{sec:prior_results} below).   Our main result is the derivation of a hierarchy of SDEs whose solutions are higher-order approximations to $q_t^m$; this hierarchy is able to achieve $O(m^{\ell/2})$ error over compact time intervals for any $\ell\in\mathbb{Z}^+$.  We derive these approximations and prove the claimed error bounds via a  bootstrapping argument.  Related techniques are commonly used in many different contexts, such as in proving regularity of solutions to various classes of equations (see page 20 in \cite{tao2006nonlinear} and page 489 in \cite{evans2010partial}),  to derive error estimates  for numerical methods 
\cite{doi:10.1080/10236198.2012.656617,10.1093/imanum/drz009} and in predictor-corrector methods (see Chapter 15.5 in \cite{kloeden2013numerical}), and in  homogenization (see Chapters 16 and 20 in \cite{pavliotis2008multiscale}) and perturbation theory (see Chapter 3.2 in \cite{murdock1999perturbations} and Chapter 6.2 - 6.3 in  \cite{smith1985singular}).

\subsection{Langevin Equation with State-Dependent Drag and Noise}
In this work, we study generalizations of \req{model_sys} that allow for  time- and state-dependent drag, noise, and external forcing:
\begin{align}
dq_t^m=&v_t^m dt,\label{q_eq0}\\
md(v^m_t)_i=&\left(-\tilde \gamma_{ik}(t,q_t^m) (v_t^m)^k+  F_i(t,q^m_t)\right)dt+\sigma_{i\rho}(t,q_t^m)dW^\rho_t,\label{v_eq}
\end{align}
where $q_t^m$ and $v_t^m$ are $\mathbb{R}^n$-valued processes. The matrix-valued function $\tilde\gamma$ will have a symmetric part, the drag matrix, and is allowed to have an antisymmetric part, coming from a possible magnetic field; see \req{tilde_gamma_def} below for details. We again alert the reader that here, and elsewhere, the superscript $m$ on vector or matrix-valued quantities denotes the value of the mass and not a component or a power.

  Except in the simplest cases, the system \req{q_eq0} - \req{v_eq} cannot be solved explicitly, and it is difficult to study numerically, especially for small values of $m$, since the velocity process $v_t^m$ diverges as $m \to 0$.  As discussed above, solutions of the homogenized SDE can serve as  approximations to (the position components of) solutions of the original system, as long as the value of $m$ is sufficiently small. The effectiveness of this has been confirmed numerically and experimentally for physically relevant values of $m$ \cite{volpe2010influence}.  However, this approximate solution is independent of $m$.  The present work improves on this, by deriving approximate position processes which are sensitive to the variation of $m$ while still not requiring one to solve the full system.  In addition, these $m$-dependent approximations are free from the type of singularity that makes the original system \req{q_eq0}-\req{v_eq} difficult to work with when $m$ is small.

  More specifically, we obtain a hierarchy of approximations $q^{\ell,m}_t$,  $\ell\in\mathbb{Z}^+$, starting with $q^{1,m}_t\equiv q_t$, where $q^{\ell,m}_t$ approximates $q_t^m$ with $O(m^{\ell/2})$ error over compact time intervals.  These processes will be constructed inductively (on $\ell$) as solutions to SDEs of the form
\begin{align}\label{hierarchy_form}
dq^{\ell,m}_t=\tilde b(t,q^{\ell,m}_t)dt+\tilde \sigma(t,q^{\ell,m}_t)dW_t+\sqrt{m}dR^{\ell-1,m}_t.
\end{align}
Here and in the following, SDEs are defined in the It{\^o} sense.

The leading order terms in \req{hierarchy_form} are given by the same drift, $\tilde b$, (including the noise-induced drift from \cite{Hottovy2014}) and diffusion, $\tilde \sigma$, that appear in the homogenized SDE for $q_t$ (see \req{q_SDE} below). The corrections are captured by the semimartingale term  $R^{\ell-1,m}_t$.  The appropriate form of $R^{\ell-1,m}_t$ will be motivated by comparing the SDE for $q_t^m$ to the SDE for the homogenized process, $q_t$, and extracting  the  error terms.

What makes the hierarchy particularly simple is that $R^{\ell-1,m}_t$ does {\em not} depend on $q^{\ell,m}_t$, but rather is an external driving semimartingale, constructed from the  approximation at the previous step, $q^{\ell-1,m}_t$ (with $R^{0,m}_t\equiv 0$).  This means that the singular nature of the $m\to 0$ limit does not complicate the limiting drift and diffusion, even for higher-order approximations.  Moreover, the presence of $m$ in the correction process, $R^{\ell-1,m}_t$, is rather benign; it primarily serves to exponentially damp out contributions from the past history of $q^{\ell-1,m}_t$.

  In Section \ref{sec:prior_results} we summarize the prior results that will be needed in this paper. Section \ref{sec:result_summary} gives a summary of the new results that will be established.  Section \ref{sec:simple_proof} contains an outline of our proof strategy, in a simplified setting, in order to highlight the key ideas.  Our new results are fully developed in Sections \ref{sec:bounded_forcing} and \ref{sec:unbounded}.  The former covers Langevin equations driven by bounded forces and the latter covers the extension to unbounded forces. 

\subsection{Homogenized Equation in the $m\to 0$ Limit: Established Results}\label{sec:prior_results}
Here we recall several previously proven results, pertaining to the Langevin equation  \req{q_eq0}-\req{v_eq}, that will be needed going forward.  Before doing so, we need to be a bit more specific about the objects appearing in \req{v_eq}. We will assume:
\begin{enumerate}
\item    $\tilde\gamma$ is constructed from  a continuous, positive definite matrix-valued drag, $\gamma$, and an antisymmetric part generated by a $C^2$ vector potential, $\psi$, as follows:
\begin{align}\label{tilde_gamma_def}
\tilde \gamma_{ik}(t,q)\equiv\gamma_{ik}(t,q) +\partial_{q^k}\psi_i(t,q)-\partial_{q^i}\psi_k(t,q).
\end{align}
\item The   diffusion, $\sigma:[0,\infty)\times\mathbb{R}^n\to\mathbb{R}^{n\times k}$ is continuous.
\item $W$ is an $\mathbb{R}^k$-valued Wiener process on  $(\Omega,\mathcal{F},\mathcal{F}_t,P)$, a filtered probability space satisfying the usual conditions \cite{karatzas2014brownian}.
\item  The total forcing is
 \begin{align}\label{eq:F_def}
F(t,q)=-\partial_t\psi(t,q)-\nabla_q V(t,q)+\tilde F(t,q),
\end{align} 
where the $C^2$ function $V$ represents an (electrostatic) potential and $\tilde F$ is a continuous external forcing.
\end{enumerate}
Our usage of electromagnetic language is due to us viewing the antisymmetric part of the drag-matrix, $\tilde\gamma$,  as begin generated by the vector-potential of an electromagnetic field, $\psi$.  If one is not interested in such a term, then our framework still allows for consideration of quite general gradient and non-gradient forces. More  assumptions on these objects will be required as we proceed; in particular the $L^p$-convergence result of Theorem \ref{convergence_theorem} will apply only to bounded forces but Theorem \ref{thm:conv_in_prob} will prove convergence in probability for a large class of unbound forces.
  
Next, define the (kinematic) momentum
\begin{align}\label{u_def}
u_t^m=mv_t^m.
\end{align}
We showed in \cite{BirrellHomogenization} that, under appropriate assumptions, there exist unique solutions $(q_t^m,u_t^m)$, $t\in[0,\infty)$, that converge to $(q_t,0)$ as $m\rightarrow 0$; $q_t$ is the solution of a homogenized limiting SDE.  The precise nature of this convergence, and the form of the SDE for $q_t$, are given below.  These results provide the foundation that we build upon in order to derive higher-order approximations.\\

\noindent { Summary of Previous Results:}

{\em
Under the  assumptions listed in  Appendix \ref{app:assump}, one has the following convergence results (see  \cite{BirrellHomogenization} for a detailed proof):  For any  $T>0$, $p>0$, $\epsilon>0$ we have
\begin{align}
 &E\left[\sup_{t\in[0,T]}\|q_t^m-q_t\|^p\right]^{1/p}=O(m^{1/2-\epsilon}),\hspace{2mm} \sup_{t\in[0,T]}E\left[\|q_t^m-q_t\|^p\right]^{1/p}=O(m^{1/2}) ,\label{results_summary1a}\\
& E\left[\sup_{t\in[0,T]}\|u_t^m\|^p\right]^{1/p}=O(m^{1/2-\epsilon}),\hspace{2mm}
\sup_{t\in[0,T]} E\left[\|u_t^m\|^p\right]^{1/p}=O(m^{1/2})\label{results_summary1b}
\end{align}
as $m\rightarrow 0$, where $q_t$ is the solution to the  SDE 
\begin{align}\label{q_SDE}
dq_t=& \tilde \gamma^{-1}(t,q_t)F(t,q_t)dt+S(t,q_t)dt+\tilde \gamma^{-1}(t,q_t)\sigma(t,q_t) dW_t.
\end{align}
$S(t,q)$ is called the {\em noise-induced drift}, see \cite{Hottovy2014,BirrellHomogenization}, and is given by
(employing the summation convention on repeated indices):
\begin{enumerate}
\item $S^i(t,q)\equiv  \partial_{q^k}(\tilde\gamma^{-1})^{ij}(t,q) \delta^{kl}G_{jl}^{rs}(t,q)\Sigma_{rs}(t,q)$,
\item $G_{ij}^{kl}(t,q)\equiv \delta^{rk}\delta^{sl}\int_0^\infty (e^{-\zeta \tilde\gamma(t,q)})_{ir} (e^{-\zeta \tilde\gamma(t,q)})_{js} d\zeta$,
\item $\Sigma_{ij}\equiv \sigma_{i\rho}\sigma_{j\xi}\delta^{\rho\xi}$.
\end{enumerate}
Here $\delta^{kl}$ denotes the Kronecker delta.  

The initial conditions are assumed to satisfy $E[\|q^m_0\|^p]<\infty$, $E[\|q_0\|^p]<\infty$, and $E[\|q_0^m-q_0\|^p]^{1/p}=O(m^{1/2})$ for all $p>0$.

The following bounds on $q_t^m$ and  $q_t$ were also shown:
\begin{align}\label{q_Lp_bound}
E\left[\sup_{t\in[0,T]}\|q^m_t\|^p\right]<\infty
,\,\,\,E\left[\sup_{t\in[0,T]}\|q_t\|^p\right]<\infty
\end{align}
for all $m>0$, $T>0$, $p>0$.}

Note that $u_t^m=O(m^{1/2})$ translates into $v_t^m=O(m^{-1/2})$. Also, in the above, we have defined the index placement on $\tilde\gamma^{-1}$ so that 
\begin{align}\label{tilde_gamma_inv_def}
(\tilde\gamma^{-1})^{ij}\tilde\gamma_{jk}=\delta^i_k,
\end{align}
and for any $v_i$ we define the contraction $(\tilde\gamma^{-1}v)^i=(\tilde\gamma^{-1})^{ij}v_j$.  Finally, note   that what we call $\tilde F$ here was simply called $F$ in \cite{BirrellHomogenization}, whereas here we use $F$ to refer to \req{eq:F_def}.

As stated previously, a comprehensive list of assumptions that guarantee the  above convergence and boundedness properties  can be found in Appendix  \ref{app:assump}.   Of particular significance, we assume that the symmetric part, $\gamma$, of $\tilde\gamma$ is positive definite; this ensures that the term involving $\gamma$ is dissipative, and  is also sufficient to ensure that $\tilde\gamma$ is invertible.

\subsection{Summary of New Results}\label{sec:result_summary}
The main content of the present work is the derivation of a hierarchy of approximating equations for the position degrees of freedom, generalizing the $O(m^{1/2})$-accurate  \req{q_SDE}, that is capable of approximating $q_t^m$ to order $O(m^{\ell/2})$ for any $\ell$. This is done in Section \ref{sec:bounded_forcing} under appropriate boundedness assumptions on the coefficients of the equation.  Specifically, in Theorem \ref{convergence_theorem} we show that for each $\ell\in \mathbb{Z}^+$ there is a family of $\mathbb{R}^n$-valued semimartingales, $R^{\ell-1,m}_t$ such that the solutions to the SDEs
\begin{align}\label{approx_hierarchy}
dq^{\ell,m}_t=&\tilde \gamma^{-1}(t,q^{\ell,m}_t)F(t,q^{\ell,m}_t)dt+S(t,q^{\ell,m}_t)dt\\
&+\tilde \gamma^{-1}(t,q^{\ell,m}_t)\sigma(t,q^{\ell,m}_t) dW_t +\sqrt{m}dR^{\ell-1,m}_t \notag
\end{align}
 satisfy
\begin{align}\label{conv_rate}
&\sup_{t\in[0,T]}E[\|q^m_t-q^{\ell,m}_t\|^p]^{1/p}=O(m^{\ell/2}),\\
&E\left[\sup_{t\in[0,T]}\|q^m_t-q^{\ell,m}_t\|^p\right]^{1/p}=O(m^{\ell/2-\epsilon})\notag
\end{align}
for all $T>0$, $p>0$, $\epsilon>0$. We call $q^{\ell,m}_t$ the solution at the $\ell$'th level of the hierarchy.

  In Section \ref{sec:unbounded} we will use the technique developed in  \cite{herzog2015small} to significantly relax the assumption of bounded forcing, while still obtaining convergence in probability:
\begin{align}\label{conv_rate2}
\lim_{m\to 0} P\left(\frac{\sup_{t\in[0,T]}\|q^m_t-q^{\ell,m}_t\|}{m^{\ell/2-\epsilon}}>\delta\right)=0
\end{align}
for all $T>0$, $\delta>0$, $\epsilon>0$, $\ell\in\mathbb{Z}^+$; see Theorem \ref{thm:conv_in_prob}.

 The hierarchy begins with $R^{0,m}_t\equiv 0$, $q^{1,m}_t\equiv q_t$, the solution to the homogenized SDE, \req{q_SDE}. We emphasize that $R^{\ell-1,m}_t$ acts as an external forcing semimartingale, and is {\em not} dependent on $q^{\ell,m}_t$.  See Chapter V in \cite{protter2013stochastic} and Appendix \ref{app:gen_SDE} below for the general theory of SDEs that include  forcing terms of this type.

  Each $R^{\ell-1,m}_t$ will be defined in terms of  $q^{\ell-1,m}|_{[0,t]}$, the approximation at the $\ell-1$st step  up to time $t$.  In fact, it will be useful to think of the $R_t^{\ell,m}$ as functions of a continuous semimartingale.  Thought of this way, they will satisfy $R_t^{\ell-1,m}=R_t^m[q^{\ell-1,m}]$. The mapping $R_t^m$ is constructed by comparing the SDE for $q_t^m$ with the homogenized SDE, \req{q_SDE}, for $q_t$ and extracting  the  error (i.e., remainder) terms.  This will be carried out in Section \ref{sec:remainder}; see Definition \ref{main_def} for the precise definition of $R_t^m[y]$.

\section{Outline of the Proof in a Simplified Setting}\label{sec:simple_proof}
   The convergence rates in \req{conv_rate} will be obtained by showing that, for an appropriate class of continuous semimartingales, $y$, the remainder terms $R_t^{m}[y]$ are  Lipschitz transformations of $y|_{[0,t]}$ (Lipschitz with respect to pairs of norms that will be specified below), and then inductively using a Gronwall's inequality argument.   

The proof of these Lipschitz properties is quite technical, and so we  first provide an outline of our argument in  the following simplified setting: Here we work in $n=1$ dimensions, and consider the SDE
\begin{align}
dq_t^m &= v_t^m\,dt, \label{eq:q_SDE_simple}\\
m\,dv_t^m &= F\left(q_t^m\right)\,dt - \gamma v_t^m\,dt + \sigma\,dW_t,\notag
\end{align}
where  $\gamma, \sigma$ are positive constants or, in terms of $u_t^m\equiv m v_t^m$:
\begin{align}
dq_t^m &= \frac{1}{m}u_t^m\,dt, \label{eq:q_SDE_simple2}\\
du_t^m &= F\left(q_t^m\right)\,dt - \frac{1}{m}\gamma u_t^m\,dt + \sigma\,dW_t.\notag
\end{align}
Solving the second equation for $\frac{1}{m}\gamma u_t^m\,dt$, substituting into the first, and rewriting in terms of 
\begin{align}
z_t^m \equiv \sqrt{m}v_t^m=u_t^m/\sqrt{m}
\end{align}
(the velocity, normalized to be of order $1$), one finds
\begin{align}
dq_t^m = \gamma^{-1}F\left(q_t^m\right)\,dt + \gamma^{-1}\sigma\,dW_t - \sqrt{m}\gamma^{-1}\,dz_t^m.
\end{align}
The last term is of order $\sqrt{m}$, so it vanishes in the limit $m \to 0$.   Indeed, it can be shown that the limit $q_t$ of the processes $q_t^m$ satisfies 
\begin{align}\label{eq:overdamped_simple}
dq_t = \gamma^{-1}F\left(q_t\right)\,dt + \gamma^{-1}\sigma\,dW_t
\end{align}
(see \cite{Nelson1967}).  This is a considerably simpler equation than the original system and in many situations its solution $q_t$  furnishes a good approximation of  $q_t^m$.  Here we are interested in a more accurate result, approximating $q_t^m$ by an $m$-dependent process that can still be obtained as a solution to a first-order SDE, albeit a somewhat more complicated one.  In fact, we will obtain a hierarchy of such equations, whose solutions will approximate $q_t^m$ to within an arbitrary power of $m$.  To implement it, we do not neglect the remainder 
\begin{align}
\sqrt{m}\gamma^{-1}\,dz_t^m.
\end{align}
Instead, we rewrite it by solving the equation for $u_t^m$ as an inhomogeneous linear equation (and multiplying the result by $\frac{1}{\sqrt{m}}$ to obtain $z_t^m$):
\begin{align}
z_t^m = { 1 \over \sqrt{m}}e^{-{\gamma \over m}t}\left(\sqrt{m}z_0^m + \int_0^te^{{\gamma \over m}s}F\left(q_s^m\right)\,ds + \int_0^t e^{{\gamma \over m}s}\sigma\,dW_s\right).
\end{align}
Substituting this expression into the remainder term, we obtain the delay equation
\begin{align}\label{eq:delay_simple}
dq^m_t = \gamma^{-1}F\left(q_t\right)\,dt  + \gamma^{-1}\sigma\,dW_t + \sqrt{m}\,dR_t^m[q^m],
\end{align}
where, for an arbitrary continuous semimartingale $y_t$, we define
\begin{align}
dR_t^m[y] = \gamma^{-1}\,dz^m_t[y]
\end{align}
with 
\begin{align}
{z}^m_t[y] = { 1 \over \sqrt{m}}e^{-{\gamma \over m}t}\left(\sqrt{m}\tilde{z}_0 + \int_0^te^{{\gamma \over m}s}F\left(y_t\right)\,ds + \int_0^t e^{{\gamma \over m}s}\sigma\,dW_s\right).
\end{align}
The equation for $q_t^m$ can now be thought of as a fixed point problem.  We will solve it iteratively, defining a sequence $q_t^{\ell,m}$ of approximate solutions, starting from $q_t^{1, m} = q_t$---the solution of the homogenized equation in the $m \to 0$ limit, and, given $q_t^{\ell-1,m}$, defining $q_t^{\ell, m}$ as the solution of the SDE
\begin{align}\label{eq:hierarchy_SDE_simple}
dq_t^{\ell,m} = \gamma^{-1}F\left(q_t^{\ell,m}\right)\,dt + \gamma^{-1}\sigma\,dW_t + \sqrt{m}dR_t^m\left[q^{\ell-1,m}\right].
\end{align}
Note that the the first two terms are  the same as in the SDE \req{eq:overdamped_simple}, while the last term is a fixed semimartingale forcing term, i.e., it does not depend on the process, $q^{\ell,m}_t$, that one is solving for.  As we will see, these two features are maintained in the general case.

We prove the claimed $O(m^{\ell/2})$  difference between $q^{\ell,m}_t$ and $q^m_t$ under a variety of norms, by using Gronwall's inequality and a bootstrapping argument. For specificity, here we outline the argument for the norm
\begin{align}
\|y-\tilde y\|_{2,T}\equiv \sup_{0\leq t\leq T} E[|y_t-\tilde y_t|^2]^{1/2}.
\end{align}

Subtracting \req{eq:delay_simple} and \req{eq:hierarchy_SDE_simple}, using the triangle inequality, and the simple bound $(a+b)^2\leq 2(a^2+b^2)$ for $a,b\geq 0$, we find
\begin{align}\label{eq:delta_q_simple}
&\|q^m-q^{\ell,m}\|_{2,T}^2\\
\leq &2 \| \int_0^t \gamma^{-1}F\left(q^m_s\right)- \gamma^{-1}F\left(q^{\ell,m}_s\right) ds\|_{2,T}^2+2{m}\|R^m[q^m]-R^m[q^{\ell-1,m}]\|_{2,T}^2.\notag
\end{align}
Suppose now that  $F$ is Lipschitz  and  $R^m[y]$ is Lipschitz in $y$ for the norm $\|\cdot\|_{2,T}$ (the latter is an oversimplification, but will allow us to illustrate the main idea without additional technical complications).  Estimating the first term in \req{eq:delta_q_simple} by using the Lipschitz property for $F$, along with the Cauchy-Schwarz inequality on the Lebesgue integral, and the second, using the Lipschitz property of $R^m$; one arrives at
\begin{align}
\|q^m-q^{\ell,m}\|_{2,T}^2\leq C\int_0^T \|q^m-q^{\ell,m}\|_{2,t}^2dt+Cm\|q^m-q^{\ell-1,m}]\|_{2,T}^2
\end{align}
for some constant $C>0$ (independent of $m$).  Gronwall's inequality then yields
\begin{align}\label{eq:error_simple}
\|q^m-q^{\ell,m}\|_{2,T}^2\leq Cm\|q^m-q^{\ell-1,m}]\|_{2,T}^2
\end{align}
for (a different) $C>0$.  Taking the square root, one finds that the iteration from $q^{\ell-1,m}$ to $q^{\ell,m}$ has improved the error by a factor of $\sqrt{m}$.  Starting with the base case  $\|q^m-q^{1,m}\|_{2,T}=\|q^m-q\|_{2,T}=O(m^{1/2})$ (this is  the error bound for the  overdamped limit, proven in \cite{BirrellHomogenization}) one  obtains the claimed error bounds  $\|q^m-q^{\ell,m}\|_{2,T}=O(m^{\ell/2})$   for all $\ell\in\mathbb{Z}^+$.

The  proof in the general case follows the above outline, but introduces several technical complications:
\begin{enumerate}
\item Working in arbitrary dimension $n\geq 1$, and with state-dependent, matrix-valued $\gamma$ and $\sigma$, complicates the derivation of the overdamped limit \req{eq:overdamped_simple}, as well as the remainder term in \req{eq:delay_simple}.  The required computations are  found in \cite{BirrellHomogenization}, but we outline them in Section \ref{sec:remainder} for completeness. 

\item We prove below that $R^m[y]$ is  Lipschitz under pairs of related norms, but not with respect to a single norm as was assumed above; see Lemmas  \ref{lemma:sup_E_R_bound} and \ref{lemma:E_sup_R_bound}.  This constitutes the greatest technical hurdle in this paper, and requires repeated use of the  inequalities collected in Appendix \ref{app:ineq}.

\item The computations in \req{eq:delta_q_simple} - \req{eq:error_simple}  must be generalized beyond the $\|\cdot\|_{2,T}$-norm, to accommodate the usage of the norm-pairs mentioned in item (2).  In addition, the stochastic integral terms will no longer cancel, due to state-dependence of $\gamma$ and $\sigma$.  These generalizations  require further use of the  inequalities from Appendix  \ref{app:ineq}; see Theorem \ref{convergence_theorem}.
\end{enumerate}

\section{Derivation of the Approximation Hierarchy for Bounded Forcing}\label{sec:bounded_forcing}
Having outlined our argument, we  now begin a detailed derivation of the hierarchy of approximating equations, in the general setting laid out in Section \ref{sec:prior_results} and under the assumptions from Appendix \ref{app:assump}; in particular, for bounded forcing.  The next two subsections lay the analytical groundwork, while the definition of the approximating hierarchy and the convergence proof are found in Section \ref{sec:conv_proof}.

\subsection{Identifying the Remainder Terms}\label{sec:remainder}
It is convenient to rewrite the system \req{q_eq0} - \req{v_eq} in terms of $u_t^m$ (see \req{u_def}):
\begin{align}
dq_t^m=&\frac{1}{m }u_t^m dt,\label{q_eq}\\
d(u^m_t)_i=&\left(-\frac{1}{m}\tilde \gamma_{ik}(t,q_t^m) (u_t^m)^k+  F_i(t,q^m_t)\right)dt+\sigma_{i\rho}(t,q_t^m)dW^\rho_t.\label{u_eq}
\end{align}
The next step is to combine the SDEs for $q^m_t$ and $u_t^m$ and decompose the result into two pieces: one that becomes the homogenized SDE, \req{q_SDE}, in the $m\to 0$ limit, and a remainder term that will motivate the definition of $R^m_t[y]$.

 \req{u_eq} is a linear equation for $u_t^m$, so the pathwise solution to
\begin{align}
\frac{d}{dt}\Phi_t^m=-\frac{1}{m}\tilde \gamma(t,q^m_t)\Phi^m_t,\hspace{2mm} \Phi^m_0=I,
\end{align}
(i.e., the fundamental-solution process; see Appendix \ref{app:fund_sol})  furnishes us with an explicit formula for $u_t^m$ in terms of $q_t^m$:
\begin{align}\label{u_sol}
u^m_t=\Phi^m_t\left(u^m_0+\int_0^t (\Phi^m_s)^{-1}F(s,q^m_s) ds+\int_0^t(\Phi^m_s)^{-1} \sigma(s,q^m_s) dW_s\right).
\end{align}

In principle, the above formula for $u_t^m$ allows one to formulate a delay equation for $q_t^m$ (i.e., with the right hand side depending on $q^m|_{[0,t]}$) by substituting \req{u_sol} into   \req{q_eq}.  However, doing so in this form does little to shed light on the behavior in the singular $m\to 0$ limit. Nevertheless, by first rewriting the equation for $q_t^m$ in an equivalent form we can turn this into a fruitful idea.

We begin by mimicking the convergence proof of $q_t^m$ to $q_t$, as found in \cite{BirrellHomogenization},  and separating the terms that survive in the $m\to 0$ limit from the remaining $O(\sqrt{m})$ error terms, which will then be used to define $R^m_t[y]$.  To make this section more self-contained, we will repeat a portion of that derivation here:

First solve \req{u_eq} for  $\frac{1}{m}u_t^mdt$ and substitute into \req{q_eq}  to obtain
\begin{align}\label{q_subs_eq}
d(q_t^m)^i=&(\tilde\gamma^{-1})^{ij}(t,q_t^m) F_j(t,q_t^m)dt+(\tilde\gamma^{-1})^{ij}(t,q_t^m)\sigma_{j\rho}(t,q_t^m)dW^\rho_t\\
&-(\tilde\gamma^{-1})^{ij}(t,q_t^m)d(u^m_t)_j.\notag
\end{align}
Integrating the last term by parts results in
\begin{align}
&-(\tilde\gamma^{-1})^{ij}(t,q_t^m)d(u^m_t)_j=- d((\tilde\gamma^{-1})^{ij}(t,q_t^m)(u^m_t)_j)\\
&+(u_t^m)_j\partial_t(\tilde\gamma^{-1})^{ij}(t,q_t^m)dt+\frac{1}{m}(u_t^m)_j(u^m_t)_k\delta^{kl}\partial_{q^l}(\tilde\gamma^{-1})^{ij}(t,q_t^m)  dt.\notag
\end{align}
From \req{results_summary1b}, we see that $u_t^m/\sqrt{m}$ is $O(1)$ as $m\to 0$, so the last term above is $O(1)$ as $m\to 0$ and must be further decomposed to identify the $q$-dependent piece that survives in the limit. To do this, use \req{u_eq} to compute
\begin{align}\label{Scott's_trick}
&d((u_t^m)_i(u_t^m)_j)=(u_t^m)_id(u_t^m)_j+(u_t^m)_jd(u_t^m)_i+d[u^m_i,u^m_j]_t\\
=&\frac{1}{m}(-(u_t^m)_i\tilde\gamma_{jk}(t,q_t^m)-(u_t^m)_j\tilde\gamma_{ik}(t,q_t^m))(u_t^m)_l \delta^{kl}dt\notag\\
&+\left((u_t^m)_iF_j(t,q^m_t)+(u_t^m)_jF_i(t,q^m_t)\right)dt\notag\\
&+(u_t^m)_i\sigma_{j\rho}(t,q_t^m)dW^\rho_t+(u_t^m)_j\sigma_{i\rho}(t,q_t^m)dW^\rho_t+\Sigma_{ij}(t,q_t^m)dt.\notag
\end{align}
 We wish to solve for $\frac{1}{m}(u_t^m)_j(u^m_t)_k dt$, hence we rewrite this as
\begin{align}\label{Lyapunov_eq}
&\frac{1}{m}(\tilde\gamma_{jk}(t,q_t^m)(u_t^m)_l(u_t^m)_i+\tilde\gamma_{ik}(t,q_t^m)(u_t^m)_l(u_t^m)_j)\delta^{kl}dt\\
=&-d((u_t^m)_i(u_t^m)_j)+\left((u_t^m)_iF_j(t,q^m_t)+(u_t^m)_jF_i(t,q^m_t)\right)dt\notag\\
&+(u_t^m)_i\sigma_{j\rho}(t,q_t^m)dW^\rho_t+(u_t^m)_j\sigma_{i\rho}(t,q_t^m)dW^\rho_t+\Sigma_{ij}(t,q_t^m)dt.\notag
\end{align}
If $\tilde\gamma$ is scalar-valued we can immediately solve for $\frac{1}{m}(u_t^m)_j(u^m_t)_k dt$.  In general, one must solve a Lyapunov equation (see Eq.(4.15) and surrounding material in \cite{BirrellHomogenization} for details). Doing so, and substituting back into \req{q_subs_eq} results in
\begin{align}\label{q_eq2}
d(q_t^m)^i=&(\tilde\gamma^{-1})^{ij}(t,q_t^m)F_j(t,q^m_t)dt+Q^{ikl}(t,q_t^m)J_{kl}(t,q_t^m)dt\\
&+(\tilde\gamma^{-1})^{ij}(t,q_t^m)\sigma_{j\rho}(t,q_t^m)dW^\rho_t+\sqrt{m}d(R^m_t)^i,\notag
\end{align}
where
\begin{align}\label{J_def}
J_{ij}(t,q)\equiv G_{ij}^{kl}(t,q)\Sigma_{kl}(t,q),
\end{align}
\begin{align}
G_{ij}^{kl}(t,q)\equiv \delta^{rk}\delta^{sl}\int_0^\infty (e^{-\zeta \tilde\gamma(t,q)})_{ir} (e^{-\zeta \tilde\gamma(t,q)})_{js} d\zeta,
\end{align}
\begin{align}\label{Q_def}
 Q^{ijl}(t,q)\equiv \partial_{q^k}(\tilde\gamma^{-1})^{ij}(t,q) \delta^{kl},
\end{align}
and, defining the $O(1)$ processes
\begin{align}\label{z_def}
z_t^m\equiv u_t^m/\sqrt{m},
\end{align}
\begin{align}\label{R_def}
d(R^m_t)^i\equiv&- d((\tilde\gamma^{-1})^{ij}(t,q_t^m)(z^m_t)_j)+(z_t^m)_j\partial_t(\tilde\gamma^{-1})^{ij}(t,q_t^m)dt\\
&+Q^{ikl}(t,q_t^m)G_{kl}^{ab}(t,q_t^m)\left((z_t^m)_a F_b(t,q^m_t)+(z_t^m)_b F_a(t,q^m_t)\right)dt\notag\\
& +(z_t^m)_a(z_t^m)_b(z_t^m)^c\partial_{q^c}(Q^{ikl}G_{kl}^{ab})(t,q_t^m) dt\notag\\
&+Q^{ikl}(t,q_t^m)G_{kl}^{ab}(t,q_t^m)\left((z_t^m)_a\sigma_{b\rho}(t,q_t^m)+(z_t^m)_b\sigma_{a\rho}(t,q_t^m)\right)dW^\rho_t\notag\\
&-\sqrt{m}d(Q^{ikl}(t,q_t^m)G_{kl}^{ab}(t,q_t^m)(z_t^m)_a(z_t^m)_b)\notag\\
&+ \sqrt{m} (z_t^m)_a(z_t^m)_b \partial_t(Q^{ikl}G_{kl}^{ab})(t,q_t^m)dt,\notag
\end{align}
with $R^m_0=0$. Though the precise form of $R^m_t$ is less than intuitive, we emphasize that it simply constitutes all of the terms in the equation for $q^m_t$ that are shown in  \cite{BirrellHomogenization} to not  contribute in the $m\to0$ limit. These terms will of course contribute to the higher-order approximations.

We have written $R^m_t$ in terms of $z_t^m$ (which is $O(1)$) so that the order in $m$ of each term is more obvious.  Note that the above definition of the remainder $R_t^m$ differs from that in \cite{BirrellHomogenization} by the factor of $\sqrt{m}$ that we have explicitly pulled out in \req{q_eq2}. Also, in obtaining \req{R_def}  we have integrated the term $-Q^{ikl}(t,q_t^m)G_{kl}^{ab}(t,q_t^m)d((u_t^m)_a(u_t^m)_b)$ from \cite{BirrellHomogenization} by parts.

One can then use \req{u_sol} to write $z_t^m$ in terms of $q_t^m$ and substitute into \req{R_def}.  In this manner, we can view   \req{q_eq2} as a delay equation for $q_t^m$, which will be the basis for the rest of the derivation.

\subsection{Lipschitz and Boundedness Properties}

As discussed above, we are viewing $z_t^m$ as defined in terms of $q^m|_{[0,t]}$ via \req{u_sol} and \req{z_def}, and similarly for $R_t^m$, by substituting the expression for $z_t^m$ into \req{R_def}.

It will be useful to view both $z_t^m$ and $R_t^m$ as functions of an arbitrary continuous semimartingale, $y$, as follows:
\begin{definition}
For $y$ a continuous semimartingale, define
\begin{align}\label{z_func_def}
&z^m_t[y]\\
\equiv&\frac{1}{\sqrt{m}}\Phi^m_t[y]\left(\sqrt{m}z^m_0+\int_0^t (\Phi^m_s[y])^{-1}F(s,y_s) ds+\int_0^t(\Phi^m_s[y])^{-1} \sigma(s,y_s) dW_s\right),\notag
\end{align}
where $\Phi_t^m[y]$ is defined pathwise as the solution to
\begin{align}
\frac{d}{dt}\Phi_t^m=-\frac{1}{m}\tilde \gamma(t,y_t)\Phi^m_t,\hspace{2mm} \Phi^m_0=I.
\end{align}
Using this, now  define 
\begin{align}\label{R_func_def}
d(R^m_t[y])^i\equiv&- d((\tilde\gamma^{-1})^{ij}(t,y_t)(z^m_t[y])_j)+(z_t^m[y])_j\partial_t(\tilde\gamma^{-1})^{ij}(t,y_t)dt\\
&+Q^{ikl}(t,y_t)G_{kl}^{ab}(t,y_t)\left((z_t^m[y])_a F_b(t,y_t)+(z_t^m[y])_b F_a(t,y_t)\right)dt\notag\\
& +(z_t^m[y])_a(z_t^m[y])_b(z_t^m[y])^c\partial_{q^c}(Q^{ikl}G_{kl}^{ab})(t,y_t) dt\notag\\
&+Q^{ikl}(t,y_t)G_{kl}^{ab}(t,y_t)\left((z_t^m[y])_a\sigma_{b\rho}(t,y_t)+(z_t^m[y])_b\sigma_{a\rho}(t,y_t)\right)dW^\rho_t\notag\\
&-\sqrt{m}d(Q^{ikl}(t,y_t)G_{kl}^{ab}(t,y_t)(z_t^m[y])_a(z_t^m[y])_b)\notag\\
&+ \sqrt{m} (z_t^m[y])_a(z_t^m[y])_b \partial_t(Q^{ikl}G_{kl}^{ab})(t,y_t)dt\notag
\end{align}
with $R^m_0[y]=0$.
\end{definition}

For any such $y$, $\Phi_t^m[y]$ is a pathwise-$C^1$ semimartingale, and $z_t^m[y]$ 
and $R_t^m[y]$ are continuous semimartingales. In terms of these maps, the processes entering the delay equation for $q_t^m$, \req{q_eq2} (see also \req{R_def}),  are given by $R_t^m=R_t^m[q^m]$ and $z_t^m=z_t^m[q^m]$, as the notation suggests.

We will denote by $Y$ the space of continuous semimartingales with
\begin{align}\label{Y_def}
E\left[\sup_{t\in [0,T]}\|y_t\|^p\right]<\infty \text{ for all $T>0$, $p>0$}.
\end{align}
We will now show that $\Phi_t^m[y]$, $z_t^m[y]$, and $R_t^m[y]$ satisfy several Lipschitz and boundedness properties for $y\in Y$.

\begin{lemma}
Let $y,\tilde y\in Y$.  Then for any  $T>0$ there exists $L_{T}>0$ such that for all $0\leq s\leq t\leq T$  we have the pathwise bound
\begin{align}\label{Phi_lip}
&\|\Phi_{t}^m[y](\Phi_s^m[y])^{-1}-\Phi_{t}^m[\tilde y](\Phi_s^m[\tilde y])^{-1}\|\\
\leq& \frac{L_T}{m} \int_s^{t} \|y_{r}-\tilde y_{r}\| dr e^{-\lambda (t-s)/m} .\notag
\end{align}

\end{lemma}
\begin{proof}
For $0\leq r\leq t-s$ define $\Phi_1(r)=\Phi_{r+s}^m[y](\Phi_s^m[y])^{-1}$ and $\Phi_2(r)=\Phi_{r+s}^m[\tilde y](\Phi_s^m[\tilde y])^{-1}$.    Lemma \ref{fund_matrix_diverg} applies to the $\Phi_i$, where $B_1(r)=-m^{-1}\tilde\gamma(r+s,y_{r+s})$, $B_2(r)=-m^{-1}\tilde\gamma(r+s,\tilde y_{r+s})$, and $\mu=-\lambda/m$,  and gives the pathwise bound
\begin{align}
&\|\Phi_{t}^m[y](\Phi_s^m[y])^{-1}-\Phi_{t}^m[\tilde y](\Phi_s^m[\tilde y])^{-1}\|\\
\leq& \frac{1}{m} \int_s^{t} \|\tilde\gamma(r,y_{r})-\tilde\gamma(r,\tilde y_{r})\| dr e^{-\lambda (t-s)/m}.\notag
\end{align}
The claimed result then follows from the fact that $\tilde\gamma$ is $C^1$ with bounded derivative on $[0,T]\times\mathbb{R}^n$ (see the assumptions in Appendix \ref{app:assump}).
\end{proof}

Next we show Lipschitz and boundedness properties for $z_t^m[y]$ and $R_t^m[y]$ under various norms.  These are the key technical results in this work; once they are established, the claimed  convergence rates in \req{conv_rate} follow from a rather straightforward application of Gronwall's inequality and a bootstrapping argument.

  In the following, we will repeatedly use the  technique of bounding a difference of products by rewriting it as
\begin{align}\label{difference property}
&ab-\tilde a\tilde b=(a-\tilde a)b+\tilde a(b-\tilde b),\\
&  abc-\tilde a\tilde b\tilde c=(a-\tilde a)bc+\tilde a(b-\tilde b)c+\tilde a\tilde b( c-\tilde c),\,\,\,\text{ etc,}\notag
\end{align}
and then employing bounds on each of the terms and their differences.  Bounds will be obtained by using the properties in Appendix \ref{app:assump} along with repeated usage of the following inequalities: H{\"o}lder's inequality (H), H{\"o}lder's inequality for finite sums (HFS), Minkowski's inequality for integrals (MI), the $L^p$-triangle inequality (T), and the the Burkholder-Davis-Gundy inequality (BDG).  These inequalities are commonly used in the literature, but  we have restated them in Appendix \ref{app:ineq} for convenience, along with textbook citations.  We will use the above abbreviations to indicate where the various inequalities are used. When multiple inequalities are used to obtain a particular line, we list the inequalities in the order they were applied.
\begin{lemma}\label{z_lip_lemma}
Let $y,\tilde y\in Y$.  Then for any $m_0>0$,  $T>0$, $q>p\geq 2$ there exist $C$, $L$ such that. for $0<m\leq m_0$, $0\leq t\leq T$:
\begin{align}
&\sup_{s\in [0,t]}E\left[\|z_s^m[y]-z_s^m[\tilde y]\|^p\right]^{1/p}\leq L\sup_{s\in [0,t]}E\left[\|y_s-\tilde y_s\|^{q}\right]^{1/{q}}
\end{align}
and
\begin{align}\label{sup_E_z_bound}
\sup_{t\in [0,T]}E\left[\|z_t^m[y]\|^p\right]\leq C.
\end{align}
We emphasize that the $C$ and $L$ are independent of $m$.
\end{lemma}
\begin{proof}
Decomposing the stochastic convolution as in \req{convol_decomp},
\begin{align}\
&\Phi^m_t[y] \int_0^t(\Phi^m_s[y])^{-1} \sigma(s,y_s) dW_s\\
=&\Phi^m_t[y] \int_0^t \sigma(s,y_s) dW_s+\frac{1}{m}\Phi^m_t[y]\int_0^t(\Phi^m_s[y])^{-1} \tilde \gamma(s,y_s) \int_s^t \sigma(r,y_r) dW_r ds,\notag
\end{align}
 we can write
\begin{align}
&z^m_t[y]\\
=&\Phi^m_t[y]z^m_0+\frac{1}{\sqrt{m}}\int_0^t \Phi^m_t[y](\Phi^m_s[y])^{-1}F(s,y_s) ds+\frac{1}{\sqrt{m}}\Phi^m_t[y] \int_0^t \sigma(s,y_s) dW_s\notag\\
&+\frac{1}{m^{3/2}}\int_0^t\Phi^m_t[y](\Phi^m_s[y])^{-1} \tilde \gamma(s,y_s) \int_s^t \sigma(r,y_r) dW_r ds.\notag
\end{align}

Let  $0\leq t\leq T$.  First, use the technique of \req{difference property} to bound the norm of the difference as follows:
\begin{align}
&\|z^m_t[y]-z^m_t[\tilde y]\|\leq\|\Phi^m_t[y]-\Phi^m_t[\tilde y]\|\|z^m_0\|\\
&+\frac{1}{\sqrt{m}}\int_0^t \|\Phi^m_t[y](\Phi^m_s[y])^{-1}- \Phi^m_t[\tilde y](\Phi^m_s[\tilde y])^{-1}\|\|F(s, y_s)\|ds\notag\\
&+\frac{1}{\sqrt{m}}\int_0^t\| \Phi^m_t[\tilde y](\Phi^m_s[\tilde y])^{-1}\|\|F(s, y_s) - F(s,\tilde y_s)\| ds\notag\\
&+\frac{1}{\sqrt{m}}\|\Phi^m_t[y]\|\| \int_0^t \sigma(s,y_s) -\sigma(s,\tilde y_s) dW_s\|\notag\\
&+\frac{1}{\sqrt{m}}\|\Phi^m_t[y] -\Phi^m_t[\tilde y]\|\| \int_0^t \sigma(s,\tilde y_s) dW_s\|\notag\\
&+\frac{1}{m^{3/2}}\!\!\int_0^t\!\!\|\Phi^m_t[y](\Phi^m_s[y])^{-1}\!\! -\Phi^m_t[\tilde y](\Phi^m_s[\tilde y])^{-1}\| \|\tilde \gamma(s, y_s) \|\|\int_s^t \!\!\sigma(r, y_r) dW_r\|ds\notag\\
&+\frac{1}{m^{3/2}}\!\!\int_0^t\!\!\|\Phi^m_t[\tilde y](\Phi^m_s[\tilde y])^{-1}\| \|\tilde \gamma(s, y_s)\|\| \int_s^t\!\! \sigma(r, y_r) - \sigma(r,\tilde y_r) dW_r\|ds\notag\\
&+\frac{1}{m^{3/2}}\!\!\int_0^t\!\!\|\Phi^m_t[\tilde y](\Phi^m_s[\tilde y])^{-1}\| \| \tilde \gamma(s, y_s) - \tilde \gamma(s,\tilde y_s)\|\| \int_s^t \!\!\sigma(r,\tilde y_r) dW_r\|ds.\notag
\end{align}

Let $L$ denote a constant, independent of $m$, that may change from line to line. Using \req{Phi_lip}, \req{fund_sol_decay}, and the bounds and Lipschitz properties from Appendix \ref{app:assump}, we obtain
\begin{align}\label{delta_z_bound}
&\|z^m_t[y]-z^m_t[\tilde y]\|\\
\leq&\frac{LL_T}{m}\int_0^t\|y_r-\tilde y_r\|dre^{-\lambda t/m}+\frac{L_T\|F\|_\infty}{m^{3/2}}\int_0^t  \int_s^{t} \|y_{r}-\tilde y_{r}\| dr e^{-\lambda (t-s)/m} ds\notag
\end{align}
\begin{align}
&+\frac{L}{\sqrt{m}}\int_0^te^{-\lambda(t-s)/m}\| y_s - \tilde y_s\| ds+\frac{e^{-\lambda t/m}}{\sqrt{m}}\| \int_0^t \sigma(s,y_s) -\sigma(s,\tilde y_s) dW_s\|\notag\\
&+\frac{L_T}{m^{3/2}} e^{-\lambda t/m}\int_0^{t} \|y_{r}-\tilde y_{r}\|  \| \int_0^t \sigma(s,\tilde y_s) dW_s\|dr\notag\\
&+\frac{L_T\|\tilde\gamma\|_\infty}{m^{5/2}}\int_0^t \int_s^{t} \|y_{r}-\tilde y_{r}\| dr e^{-\lambda (t-s)/m} \|\int_s^t \sigma(r, y_r) dW_r\|ds\notag\\
&+\frac{\|\tilde\gamma\|_\infty}{m^{3/2}}\int_0^te^{-\lambda(t-s)/m} \| \int_s^t \sigma(r, y_r) - \sigma(r,\tilde y_r) dW_r\|ds\notag\\
&+\frac{L}{m^{3/2}}\int_0^te^{-\lambda(t-s)/m} \|  y_s - \tilde y_s\|\| \int_s^t \sigma(r,\tilde y_r) dW_r\|ds.\notag
\end{align}
Here, $\|g\|_\infty\equiv \sup_{t\in[0,T],q\in\mathbb{R}^n}\|g(t,q)\|$.  

Next we  compute an $L^p$ bound  by using the inequalities from Appendix \ref{app:ineq} (indicated by the abbreviations in parentheses; if more than one inequality is used, we list them in order of usage): Let $q>p\geq 2$ and define $\tilde p=pq/(q-p)$. First use the triangle inequality (T) to bound the $L^p$-norm of $\|z^m_t[y]-z^m_t[\tilde y]\|$ by the sum of $L^p$-norms of each of the terms on the right-hand-side in \req{delta_z_bound}, and then bound each term as follows:
\begin{align}
&E[\|z^m_t[y]-z^m_t[\tilde y]\|^p]^{1/p}\\
&\leq\frac{LL_T}{m}e^{-\lambda t/m}\int_0^tE\left[\|y_r-\tilde y_r\|^p\right]^{1/p}dr&\text{(MI)}\notag\\
&+\frac{L_T\|F\|_\infty}{m^{3/2}}\int_0^t  \int_s^{t} e^{-\lambda (t-s)/m}E\left[\|y_{r}-\tilde y_{r}\|^p\right]^{1/p}drds&\text{(MI)}\notag\\
&+\frac{L}{\sqrt{m}}\int_0^te^{-\lambda(t-s)/m}E\left[\| y_s - \tilde y_s\|^p\right]^{1/p}ds&\text{(MI)}\notag\\
&+\frac{1}{\sqrt{m}}e^{-\lambda t/m}E\left[\| \int_0^t \sigma(s,y_s) -\sigma(s,\tilde y_s) dW_s\|^p\right]^{1/p}\notag\\
&+\frac{L_T}{m^{3/2}} e^{-\lambda t/m}\int_0^{t}E\left[ \|y_{r}-\tilde y_{r}\|^q\right]^{1/q}E\left[ \| \int_0^t \sigma(s,\tilde y_s) dW_s\|^{\tilde p}\right]^{1/\tilde p}\!\!\!\!\!dr&\text{(MI,H)} \notag\\
&+\frac{L_T\|\tilde\gamma\|_\infty}{m^{5/2}}\int_0^t \int_s^{t} e^{-\lambda (t-s)/m}E\left[ \|y_{r}-\tilde y_{r}\|^q\right]^{1/q}&\text{(MI,H)}\notag\\
&\hspace{30mm}\times E\left[  \|\int_s^t \sigma(\tilde r, y_{\tilde r}) dW_{\tilde r}\|^{\tilde p}\right]^{1/\tilde p}\!\!\!\!\!drds\notag
\end{align}
\begin{align}
&+\frac{\|\tilde\gamma\|_\infty}{m^{3/2}}\int_0^t e^{-\lambda(t-s)/m}E\left[ \| \int_s^t \sigma(r, y_r) - \sigma(r,\tilde y_r) dW_r\|^p\right]^{1/p}\!\!\!\!\!ds &\text{(MI)}\notag\\
&+\frac{L}{m^{3/2}}\int_0^t\!e^{-\lambda(t-s)/m}E\left[ \|  y_s - \tilde y_s\|^q\right]^{1/q}E\left[\| \int_s^t \sigma(r,\tilde y_r) dW_r\|^{\tilde p}\right]^{1/\tilde p}\!\!\!\!\!ds.&\text{(MI,H)}\notag
\end{align}
The  uses of H{\"o}lder's inequality are all with the conjugate exponents $q/p$ and $q/(q-p)$.

Bounding the stochastic integrals via the Burkholder-Davis-Gundy inequality, using the bounds and Lipschitz properties form Appendix   \ref{app:assump}, and extracting the powers of $m$ from the Lebesgue integrals  gives
\begin{align}
E[\|z^m_t[y]-&z^m_t[\tilde y]\|^p]^{1/p}\leq L\frac{ t}{m}e^{-\lambda t/m} \sup_{r\in[0,t]}E\left[\|y_r-\tilde y_r\|^p\right]^{1/p}\\
&+Lm^{1/2}\int_0^{t/m}  u e^{-\lambda u}du  \sup_{r\in[0,t]}E\left[\|y_r-\tilde y_r\|^p\right]^{1/p}\notag\\  
&+L{m}^{1/2}\int_0^{t/m} e^{-\lambda u} du \sup_{s\in[0,t]}E\left[\|y_s-\tilde y_s\|^p\right]^{1/p}\notag\\
&+\frac{1}{\sqrt{m}}e^{-\lambda t/m}E\left[\left( \int_0^t \|y_s -\tilde y_s\|^2ds\right)^{p/2}\right]^{1/p}&\text{(BDG)}\notag\\
&+L(t/m)^{3/2} e^{-\lambda t/m} \sup_{r\in[0,t]}E\left[ \|y_{r}-\tilde y_{r}\|^q\right]^{1/q}&\text{(BDG)}\notag\\
&+L\int_0^{t/m} u^{3/2} e^{-\lambda u}du\sup_{r\in[0,t]}E\left[ \|y_{r}-\tilde y_{r}\|^q\right]^{1/q}&\text{(BDG)}\notag\\
&+\frac{L}{m^{3/2}}\int_0^t e^{-\lambda(t-s)/m}E\left[\left( \int_s^t \| y_r - \tilde y_r\|^2 dr\right)^{p/2}\right]^{1/p}ds&\text{(BDG)}\notag\\
&+L\int_0^{t/m}e^{-\lambda u} u^{1/2}du\sup_{s\in[0,t]}E\left[ \|  y_s - \tilde y_s\|^q\right]^{1/q}&\text{(BDG)}\notag\\
\leq& L \sup_{s\in[0,t]}E\left[ \|y_{s}-\tilde y_{s}\|^q\right]^{1/q}.&\text{(H)}\notag
\end{align}
We have used several times the fact that $\sup_{m>0}\sup_{t\geq 0} (t/m)^ke^{-\lambda  t/m}<\infty$ for all $k\geq 0$.  To obtain the last line, we used H{\"o}lder's inequality to  bound all expectations by the $L^q$-norm; the condition $p\geq 2$ was needed to use H\"older's inequality on the term in the third-to-last line. Taking a supremum over $t$ on the left hand side gives the claimed Lipschitz bound.

The bound on $z^m_t$ proceeds using the same tools. First we bound
\begin{align}\label{norm_z_bound}
&\|z^m_t[y]\|\\
\leq&Ce^{-\lambda t/m}+\frac{\|F\|_\infty}{\sqrt{m}}\int_0^t e^{-\lambda(t-s)/m} ds+\frac{1}{\sqrt{m}}e^{-\lambda t/m} \|\int_0^t \sigma(s,y_s) dW_s\|\notag\\
&+\frac{\|\tilde \gamma\|_\infty}{m^{3/2}}\int_0^t e^{-\lambda(t-s)/m} \|\int_s^t \sigma(r,y_r) dW_r\| ds.\notag
\end{align}
Hence (letting $C$ vary from line to line) for $p\geq 2$ we obtain
\begin{align}
&E[\|z^m_t[y]\|^p]^{1/p}\\
\leq& Ce^{-\lambda t/m}+\|F\|_\infty\sqrt{m}\int_0^{t/m}\!\!\!\! e^{-\lambda u} du&\text{(T)}\notag\\
&+\frac{1}{\sqrt{m}}e^{-\lambda t/m} E\left[\|\int_0^t \sigma(s,y_s) dW_s\|^p\right]^{1/p}&\text{(T)}\notag\\
&+\frac{\|\tilde \gamma\|_\infty}{m^{3/2}}E\left[\left(\int_0^t e^{-\lambda(t-s)/m} \|\int_s^t \sigma(r,y_r) dW_r\| ds\right)^p\right]^{1/p}&\text{(T)}\notag\\
\leq&C+\frac{C}{\sqrt{m}}e^{-\lambda t/m} E\left[\left(\int_0^t \|\sigma(s,y_s)\|^2 ds\right)^{p/2}\right]^{1/p}&\text{(BDG)}\notag\\
&+\frac{C}{m^{3/2}}\int_0^te^{-\lambda(t-s)/m}E\left[\left(  \int_s^t \|\sigma(r,y_r)\|^2dr \right)^{p/2}\right]^{1/p}ds.&\text{(MI,BDG)}\notag
\end{align}
We have assumed $\sigma$ is bounded, therefore
\begin{align}
&E[\|z^m_t[y]\|^p]^{1/p}\\
\leq&C+C(t/m)^{1/2}e^{-\lambda t/m} +\frac{C}{m^{3/2}}\int_0^te^{-\lambda(t-s)/m}(t-s)^{1/2}ds\notag\\
\leq&C+C(t/m)^{1/2}e^{-\lambda t/m}  +C\int_0^{t/m}e^{-\lambda u}u^{1/2}du.\notag
\end{align}
Taking the supremum over $t\in[0,T]$, we arrive at the claimed bound \req{sup_E_z_bound}.
\end{proof}

Similarly to the previous lemma, but this time also employing Lemma  \ref{Phi_int_bound}, we have:
\begin{lemma}\label{lemma:E_sup_z}
Let $y,\tilde y\in Y$.  Then for any $m_0>0$,  $T>0$, $p\geq 2$, $\epsilon>0$ there exist $q>p$, $C$, $L$ such that for $0<m\leq m_0$, $0\leq t\leq T$:
\begin{align}\label{E_sup_lipschitz_bound}
&E\left[\sup_{s\in [0,t]}\|z_s^m[y]-z_s^m[\tilde y]\|^p\right]^{1/p}\leq \frac{L}{m^\epsilon}E\left[\sup_{s\in [0,t]}\|y_s-\tilde y_s\|^{q}\right]^{1/{q}}
\end{align}
and
\begin{align}\label{E_sup_z_bound}
E\left[\sup_{t\in [0,T]}\|z_t^m[y]\|^p\right]^{1/p}\leq C/m^\epsilon.
\end{align}
Again, $C$ and $L$ are independent of $m$.  Also note that, in contrast with the previous result, $q$ depends on $\epsilon$ and can't be chosen arbitrarily.
\end{lemma}
\begin{proof}
Let $C$ be a constant that varies from line to line. To derive  \req{E_sup_z_bound} we start from \req{norm_z_bound} and use the triangle inequality (T) to compute 
 \begin{align}\label{eq:E_sup_z}
&E\left[\sup_{t\in[0,T]}\|z^m_t[y]\|^p\right]^{1/p}\!\!\!\!\!\!
\leq C+\frac{1}{\sqrt{m}}E\left[\left(\sup_{t\in[0,T]}e^{-\lambda t/m} \|\int_0^t \sigma(s,y_s) dW_s\|\right)^p\right]^{1/p}\\
&+\frac{C}{m^{3/2}}E\left[\left(\sup_{t\in[0,T]}\int_0^t e^{-\lambda(t-s)/m} \|\int_s^t \sigma(r,y_r) dW_r\| ds\right)^p\right]^{1/p}.\notag
\end{align}

Here, and in the following, we will need to bound expected values of the following types:
\begin{align}\label{E1_def}
E_{1,j}\equiv E\left[\left(\sup_{t\in[0,T]}t^je^{-\lambda t/m} \|\int_0^t V_s dW_s\|\right)^p\right]^{1/p}
\end{align}
and
\begin{align}\label{E2_def}
E_{2,j}\equiv E\left[\left(\sup_{t\in[0,T]}\int_0^t (t-s)^je^{-\lambda(t-s)/m} \|\int_s^t V_r dW_r\| ds\right)^p\right]^{1/p},
\end{align}
where $j\in\mathbb{Z}_0$ and $V_s$ is a continuous adapted $\mathbb{R}^{n\times k}$-valued process.

First, for any $T>0$, $p\geq 2$,  use integration by parts to write
\begin{align}
&E_{1,j}= E\left[\left(\sup_{t\in[0,T]} \!\|\int_0^t \!s^je^{-\lambda s/m} V_s dW_s+\!\int_0^t  \int_0^s V_rdW_r \partial_s(s^je^{-\lambda s/m})ds\|\right)^p\right]^{1/p}
\end{align}
and then use the inequalities from Appendix \ref{app:ineq}  to obtain
\begin{align}
E_{1,j}\leq& E\left[\left(\sup_{t\in[0,T]} \|\int_0^t s^je^{-\lambda s/m} V_s dW_s\|\right)^p\right]^{1/p}&\text{(T)}\notag\\
&+\int_0^TE\left[  \|\int_0^s V_rdW_r\|^p \right]^{1/p}h_{j,m}(s)ds&\text{(T,MI)}\notag\\
\leq& C_{p,n}^{1/p}\left(\int_0^T E\left[  \|V_s\|^{p}\right]^{2/p}s^{2j}e^{-2\lambda s/m}ds\right)^{1/2}&\text{(BDG,MI)}\notag\\
&+C_{p,n}^{1/p}\int_0^T\!\!\left(\int_0^s E[\|V_r\|^p]^{2/p}dr\right)^{1/2}\!\!\!h_{j,m}(s)ds&\text{(BDG,MI)}\notag\\
\leq& C_{1,p,n,j,\lambda}m^{j+1/2}\sup_{s\in[0,T]}E\left[  \|V_s\|^{p}\right]^{1/p},\label{eq:E1_bound}
\end{align}
\begin{align}
&h_{j,m}(s)\equiv js^{j-1}e^{-\lambda s/m}+s^je^{-\lambda s/m}\lambda/m,\notag\\
&C_{1,p,n,j,\lambda}\equiv C_{p,n}^{1/p}\!\left(\!\left(\int_0^\infty\!\!\! u^{2j}e^{-2\lambda u}du\right)^{1/2}\!\!\!\!\!\!+\!\int_0^\infty \!\!\!u^{1/2}(ju^{j-1}e^{-\lambda u}+\lambda u^je^{-\lambda u})du\!\right).\notag
\end{align}
Here, $C_{p,n}$ denotes the constant from  (the vector version of) the Burkholder-Davis-Gundy inequality. Note that $p\geq 2$ was needed for the second and third usage of Minkowski's inequality for integrals.

\req{E2_def} can be bounded similarly, but this time we also need to employ the bound from Lemma \ref{Phi_int_bound}.  Using this, and then the triangle inequality (T), for any $T>0$, $\delta>0$, $p\geq 2$ we have
\begin{align}
E_{2,j}\leq&\frac{C_jm^{j+1}}{\lambda^{j+1}} \bigg(E\left[\max_{\ell=1,...,N}\sup_{\tau\in[t_{\ell-1},T_{\ell+1}]}\|\int_{t_{\ell-1}}^{\tau}V_rdW_r\|^p\right]^{1/p}\\
&+e^{-\lambda \delta/(2m)}E\left[  \sup_{\tau\in[0,T]}\|\int_0^\tau V_r dW_r\|^p\right]^{1/p}\bigg),\notag
\end{align} 
where $N\equiv\min\{\ell\in\mathbb{Z}:\ell\delta\geq T\}$, $t_{\ell-1}\equiv (\ell-1)\delta$, $T_{\ell+1}\equiv \min\{(\ell+1)\delta,T\}$, and $C_j$ depends only on $j$.  

For any $\tilde q\geq 1$ we can bound the maximum of an $N$-term sequence by its $\ell^{\tilde q}$ norm.  This, along with the inequalities from Appendix \ref{app:ineq}, yields 
\begin{align}
&E_{2,j}\\
\leq&\frac{C_jm^{j+1}}{\lambda^{j+1}} \left(\!E\left[\left(\sum_{\ell=1}^N\sup_{\tau\in[t_{\ell-1},T_{\ell+1}]}\|\int_{t_{\ell-1}}^{\tau}\!\!\!\!\!\!V_rdW_r\|^{p\tilde q}\right)^{1/\tilde q}\right]^{1/p}\right.\notag\\
&\left.+C_{p,n}^{1/p}e^{-\lambda \delta/(2m)}E\left[  \left(\int_0^T \|V_r\|^2dr\right)^{p/2}\right]^{1/p}\right)&\text{(BDG)}\notag\\
\leq&\frac{C_jm^{j+1}}{\lambda^{j+1}} \left(\!\left(\sum_{\ell=1}^NE\left[\sup_{\tau\in[t_{\ell-1},T_{\ell+1}]}\|\int_{t_{\ell-1}}^{\tau}\!\!\!\!\!\!V_rdW_r\|^{p\tilde q}\right]\right)^{1/(p\tilde q)}\right.&\text{(H)}\notag\\
&\left.+C_{p,n}^{1/p}e^{-\lambda \delta/(2m)}\left(\int_0^TE\left[   \|V_r\|^{p}\right]^{2/p}dr\right)^{1/2}\right)&\text{(MI)}\notag\\
\leq&\frac{C_jm^{j+1}}{\lambda^{j+1}}\! \left(\!C_{p\tilde q,n}^{1/(p\tilde q)}\left(\sum_{\ell=1}^NE\left[\left(\int_{t_{\ell-1}}^{T_{\ell+1}}\|V_r\|^2dr\right)^{p\tilde q/2}\right]\right)^{1/(p\tilde q)}\right.&\text{(BDG)}\notag\\
&+T^{1/2}C_{p,n}^{1/p}e^{-\lambda \delta/(2m)}\sup_{r\in[0,T]}E\left[   \|V_r\|^{p}\right]^{1/p}]\bigg)\notag\\
\leq&\frac{C_jm^{j+1}}{\lambda^{j+1}}\! \left(C_{p\tilde q,n}^{1/(p\tilde q)}\left((1+T/\delta )\left(2\delta\right)^{p\tilde q/2}\right)^{1/(p\tilde q)}\!\!+T^{1/2}C_{p,n}^{1/p}e^{-\lambda \delta/(2m)}\right)&\text{(MI,H)}\notag\\
&\times \sup_{r\in[0,T]}E\left[\|V_r\|^{p\tilde q}\right]^{1/(p\tilde q)}.\notag
\end{align} 
In obtaining the last inequality, we used $N\leq 1+T/\delta$. 

For the current purposes, it is useful to let $\delta=m^\kappa$ for $\kappa\in (0,1)$.  Hence we have shown there is a constant $C_{T,m_0,p,\kappa,\tilde q,n,\lambda,j}$, with the indicated dependencies, such that
\begin{align}
E_{2,j}\leq&\frac{C_jm^{j+1}}{\lambda^{j+1}} \left(2^{1/2}C_{p\tilde q,n}^{1/(p\tilde q)}(m^\kappa+T)^{1/(p\tilde q)}m^{\kappa(1/2-1/(p\tilde q))}\right.\notag\\
&\left.+T^{1/2}C_{p,n}^{1/p}e^{-\lambda /(2m^{1-\kappa})}\right)\sup_{r\in[0,T]}E\left[\|V_r\|^{p\tilde q}\right]^{1/(p\tilde q)}\notag\\
\leq&C_{T,m_0,p,\kappa,\tilde q,n,\lambda,j} m^{j+1+\kappa(1/2-1/(p\tilde q))}\sup_{r\in[0,T]}E\left[\|V_r\|^{p\tilde q}\right]^{1/(p\tilde q)}.\label{eq:E2_bound}
\end{align}
We note that the constant  can be chosen to be increasing in $T$.

Applying the bounds on $E_{1,j}$ and $E_{2,j}$ to \req{eq:E_sup_z} implies that, for any $\tilde q\geq 1$, $\kappa\in(0,1)$, we have
 \begin{align}
&E\left[\sup_{t\in[0,T]}\|z^m_t[y]\|^p\right]^{1/p}\leq C+ C_{1,p,n,0,\lambda}\sup_{s\in[0,T]}E[\|\sigma(s,y_s)\|^p]^{1/p}\\
&+CC_{T,m_0,p,\kappa,\tilde q,n,\lambda,0} m^{-(1/2-\kappa(1/2-1/(p\tilde q)))}\sup_{r\in[0,T]}E\left[\|\sigma(r,y_r)\|^{p\tilde q}\right]^{1/(p\tilde q)}.\notag
\end{align}
$\sigma$ is bounded, so for any $\epsilon>0$ we can fix $\kappa$ and $\tilde q$ and find a $C$ so that
\begin{align}
&E\left[\sup_{t\in[0,T]}\|z^m_t[y]\|^p\right]^{1/p}\leq C/m^\epsilon
\end{align}
as claimed.

For the Lipschitz bound we go back to \req{delta_z_bound} and compute
\begin{align}
&\|z^m_t[y]-z^m_t[\tilde y]\|\\
\leq&\bigg(\frac{LL_T}{m}te^{-\lambda t/m}+\frac{L_T\|F\|_\infty}{m^{3/2}}\int_0^t  (t-s)e^{-\lambda (t-s)/m} ds+\frac{L}{\sqrt{m}}\int_0^te^{-\lambda(t-s)/m} ds\notag\\
&+\frac{L_T}{m^{3/2}} te^{-\lambda t/m} \| \int_0^t \sigma(s,\tilde y_s) dW_s\|\notag\\
&+\frac{L_T\|\tilde\gamma\|_\infty}{m^{5/2}}\int_0^t (t-s) e^{-\lambda (t-s)/m} \|\int_s^t \sigma(r, y_r) dW_r\|ds\notag\\
&+\frac{L}{m^{3/2}}\int_0^te^{-\lambda(t-s)/m} \| \int_s^t \sigma(r,\tilde y_r) dW_r\|ds\bigg)\sup_{r\in[0,t]}\|y_r-\tilde y_r\|\notag\\
&+\frac{\|\tilde\gamma\|_\infty}{m^{3/2}}\int_0^te^{-\lambda(t-s)/m} \| \int_s^t \sigma(r, y_r) - \sigma(r,\tilde y_r) dW_r\|ds\notag\\
&+\frac{1}{\sqrt{m}}e^{-\lambda t/m}\| \int_0^t \sigma(s,y_s) -\sigma(s,\tilde y_s) dW_s\|.\notag
\end{align}

Let  $\tilde t\leq T$, take the supremum $[0,\tilde t]$ and use the triangle inequality (T) to bound the $L^p$-norm of the sum. Next, let  $q>p$,  define $\tilde p=pq/(q-p)$, and use H{\"o}lder's inequality (H) with exponents $q/p$ and $q/(q-p)$ followed by the triangle inequality (T) with exponent $\tilde p$ (again, letting $L$ vary line to line):
\begin{align}
&E\left[\sup_{t\in[0,\tilde t]}\|z^m_t[y]-z^m_t[\tilde y]\|^p\right]^{1/p}\\
\leq&\bigg(L+\frac{L}{m^{3/2}} E\bigg[\bigg(\sup_{t\in[0,\tilde t]}te^{-\lambda t/m} \| \int_0^t \sigma(s,\tilde y_s) dW_s\|\bigg)^{\tilde p}\bigg]^{1/\tilde p}&\!\!\!\!\!\!\!\!\!\text{(H,T)}\notag\\
&+\frac{L}{m^{5/2}}E\bigg[\bigg(\sup_{t\in[0,\tilde t]}\int_0^t (t-s) e^{-\lambda (t-s)/m} \|\int_s^t \sigma(r, y_r) dW_r\|ds\bigg)^{\tilde p}\bigg]^{1/\tilde p}&\!\!\!\!\!\!\!\!\!\text{(H,T)}\notag\\
&+\frac{L}{m^{3/2}}E\bigg[\bigg(\sup_{t\in[0,\tilde t]}\int_0^te^{-\lambda(t-s)/m} \| \int_s^t \sigma(r,\tilde y_r) dW_r\|ds\bigg)^{\tilde p}\bigg]^{1/\tilde p}\bigg)&\!\!\!\!\!\!\!\!\!\text{(H,T)}\notag\\
&\times E\left[\sup_{r\in[0,\tilde t]}\|y_r-\tilde y_r\|^q\right]^{1/q}&\!\!\!\!\!\!\!\!\!\text{(H)}\notag\\
&+\frac{L}{m^{3/2}}\!E\left[\left(\sup_{t\in[0,\tilde t]}\int_0^te^{-\lambda(t-s)/m} \| \int_s^t \sigma(r, y_r) - \sigma(r,\tilde y_r) dW_r\|ds\right)^p\right]^{1/p}\notag\\
&+\frac{1}{\sqrt{m}}E\!\left[\left(\sup_{t\in[0,\tilde t]}e^{-\lambda t/m}\| \int_0^t \sigma(s,y_s) -\sigma(s,\tilde y_s) dW_s\|\right)^p\right]^{1/p}.\notag
\end{align}

Each of the expected values involving $\sigma$ is of the form \req{E1_def} or \req{E2_def}.  Applying the corresponding bounds, \req{eq:E1_bound} and \req{eq:E2_bound},  to the first term results in
\begin{align}
&\frac{1}{m^{3/2}} E\bigg[\bigg(\sup_{t\in[0,\tilde t]}te^{-\lambda t/m} \| \int_0^t \sigma(s,\tilde y_s) dW_s\|\bigg)^{\tilde p}\bigg]^{1/\tilde p}\\
&+\frac{1}{m^{5/2}}E\bigg[\bigg(\sup_{t\in[0,\tilde t]}\int_0^t (t-s) e^{-\lambda (t-s)/m} \|\int_s^t \sigma(r, y_r) dW_r\|ds\bigg)^{\tilde p}\bigg]^{1/\tilde p}\notag\\
&+\frac{1}{m^{3/2}}E\bigg[\bigg(\sup_{t\in[0,\tilde t]}\int_0^te^{-\lambda(t-s)/m} \| \int_s^t \sigma(r,\tilde y_r) dW_r\|ds\bigg)^{\tilde p}\bigg]^{1/\tilde p}\notag\\
\leq& \|\sigma\|_\infty C_{1,pq/(q-p),n,1,\lambda}\label{coeff_1}\\
&+\|\sigma\|_\infty C_{T,m_0,\tilde p,\kappa,\tilde q,n,\lambda,1} m^{-5/2+2+\kappa(1/2-(q-p)/(\tilde q pq))}\notag\\
&+\|\sigma\|_\infty C_{T,m_0,{\tilde p},\kappa,\tilde q,n,\lambda,0} m^{-3/2+1+\kappa(1/2-(q-p)/(\tilde q pq))}\notag
\end{align}
for any $\tilde q\geq 1$, $\kappa\in(0,1)$. 

 Therefore, given $\epsilon>0$ we can choose $\tilde q$, $\kappa$ 
 \begin{align}
&E\left[\sup_{t\in[0,\tilde t]}\|z^m_t[y]-z^m_t[\tilde y]\|^p\right]^{1/p}\leq \frac{L}{m^{\epsilon}} E\left[\sup_{r\in[0,\tilde t]}\|y_r-\tilde y_r\|^q\right]^{1/q}\\
&+\frac{L}{m^{3/2}}E\left[\left(\sup_{t\in[0,\tilde t]}\int_0^te^{-\lambda(t-s)/m} \| \int_s^t \sigma(r, y_r) - \sigma(r,\tilde y_r) dW_r\|ds\right)^p\right]^{1/p}\notag\\
&+\frac{1}{\sqrt{m}}E\left[\left(\sup_{t\in[0,\tilde t]}e^{-\lambda t/m}\| \int_0^t \sigma(s,y_s) -\sigma(s,\tilde y_s) dW_s\|\right)^p\right]^{1/p}\notag
\end{align}
for some $L>0$.

We can similarly use \req{eq:E1_bound} and \req{eq:E2_bound} to bound the last two terms:
\begin{align}
&E\left[\sup_{t\in[0,\tilde t]}\|z^m_t[y]-z^m_t[\tilde y]\|^p\right]^{1/p}\leq \frac{L}{m^{\epsilon}} E\left[\sup_{r\in[0,\tilde t]}\|y_r-\tilde y_r\|^q\right]^{1/q}\\
&+L C_{T,m_0,p,\kappa,\tilde q,n,\lambda,0} m^{-3/2+1+\kappa(1/2-1/(p\tilde q))}\notag\\
&\times \sup_{r\in[0,\tilde t]}E\left[\|\sigma(r, y_r) - \sigma(r,\tilde y_r)\|^{p\tilde q}\right]^{1/(p\tilde q)}\notag\\
&+C_{1,p,n,0,\lambda}\sup_{s\in[0,\tilde t]}E\left[  \|\sigma(s,y_s) -\sigma(s,\tilde y_s)\|^{p}\right]^{1/p}.\notag
\end{align}
Using the  assumption that $\sigma$ is Lipschitz in its second argument, for appropriate choices of $\kappa\in(0,1)$ and $\tilde q\geq 1$  we finally obtain
\begin{align}
&E\left[\sup_{t\in[0,\tilde t]}\|z^m_t[y]-z^m_t[\tilde y]\|^p\right]^{1/p}\leq \frac{L}{m^{\epsilon}} E\left[\sup_{s\in[0,\tilde t]}\|y_s-\tilde y_s\|^{p\tilde q}\right]^{1/(p\tilde q)}.
\end{align}
This completes the proof.
\end{proof}

Next we prove analogous results for $R^m_t[y]$.
\begin{lemma}\label{lemma:sup_E_R_bound}
Let $y,\tilde y\in Y$.  Then for any $m_0>0$,  $T>0$, $q>p\geq 2$, there exist  $C$, $L$ such that for $0<m\leq m_0$, $0\leq t\leq T$:
\begin{align}
&\sup_{s\in [0,t]}E\left[\|R_s^m[y]-R_s^m[\tilde y]\|^p\right]^{1/p}\leq L\sup_{s\in [0,t]}E\left[\|y_s-\tilde y_s\|^{q}\right]^{1/{q}}.
\end{align}
and
\begin{align}
\sup_{t\in[0,T]}E\left[\|R_t^m[y]\|^p\right]\leq C.
\end{align}
Once again,  $C$ and $L$ are independent of $m$.
\end{lemma}
\begin{proof}
We group the terms in $R_t^m[y]$ of similar form as follows:
\begin{align}\label{eq:R_grouping}
&(R^m_t[y])^i\\
=&- ((\tilde\gamma^{-1})^{ij}(t,y_t)(z^m_t[y])_j-(\tilde\gamma^{-1})^{ij}(0,y_0)(z^m_0[y])_j)\notag\\
&-\sqrt{m}\left((Q^{ikl}G_{kl}^{ab})(t,y_t)(z_t^m[y])_a(z_t^m[y])_b\!-\!(Q^{ikl}G_{kl}^{ab})(0,y_0)(z_0^m[y])_a(z_0^m[y])_b\right)\notag\\
&+\int_0^t(z_s^m[y])_j\left(\partial_s(\tilde\gamma^{-1})^{ij}+Q^{ikl}G_{kl}^{jb}F_b+Q^{ikl}G_{kl}^{aj} F_a\right)(s,y_s)ds\notag\\
& +\int_0^t(z_s^m[y])_a(z_s^m[y])_b(z_s^m[y])^c\partial_{q^c}(Q^{ikl}G_{kl}^{ab})(s,y_s) ds\notag\\
&+ \sqrt{m} \int_0^t(z_s^m[y])_a(z_s^m[y])_b \partial_s(Q^{ikl}G_{kl}^{ab})(s,y_s)dt\notag\\
&+\int_0^t (z_s^m[y])_j \left(Q^{ikl}G_{kl}^{jb}\sigma_{b\rho}+Q^{ikl}G_{kl}^{aj}\sigma_{a\rho}\right)(s,y_s)dW^\rho_s.\notag
\end{align}
We will show that each of these terms satisfies the claimed Lipschitz property.     The computations are all similar, and use the same tools as the previous lemmas, so we  illustrate the main ideas while omitting some details:

We will again make repeated use of the \req{difference property}. The estimates of the first two terms are similar, and the derivations rely on the fact that $\tilde\gamma^{-1}$, $Q^{ikl}(t,q)$, and $G_{kl}^{ab}(t,q)$ are bounded and Lipschitz in $q$, uniformly in $t\in[0,T]$, and $z_0^m[y]=z_0^m$ is independent of $y$ and  uniformly bounded (these properties follow from the assumptions in Appendix \ref{app:assump}; boundedness of $G_{kl}^{ab}$ follows from the lower bound on the eigenvalues of $\gamma$). We illustrate the computation with the first sub-term of the second term. Again, let $L$ denote a constant that may vary from line to line and use the triangle inequality (T) to bound the $L^p$-norm as follows:
\begin{align}
&\sup_{t\in[0,\tilde t]}E\left[\|\sqrt{m}(Q^{ikl}G_{kl}^{ab})(t,y_t)(z_t^m[y])_a(z_t^m[y])_b\right.\\
&\left.-\sqrt{m}(Q^{ikl}G_{kl}^{ab})(t,\tilde y_t)(z_t^m[\tilde y])_a(z_t^m[\tilde y])_b\|^p\right]^{1/p}\notag\\
\leq&\sqrt{m}\sup_{t\in[0,\tilde t]}E\left[\|(Q^{ikl}G_{kl}^{ab})(t,y_t)((z_t^m[y])_a(z_t^m[y])_b-(z_t^m[\tilde y])_a(z_t^m[\tilde y])_b)\|^p\right]^{1/p}\notag\\
&+\!\sqrt{m}\sup_{t\in[0,\tilde t]}\!\!E\left[\|((Q^{ikl}G_{kl}^{ab})(t, y_t)\!-\!(Q^{ikl}G_{kl}^{ab})(t,\tilde y_t))(z_t^m[\tilde y])_a(z_t^m[\tilde y])_b\|^p\right]^{1/p}\notag\\
\leq&L\bigg(\sup_{t\in[0,\tilde t]}\!\!E\left[\|z_t^m[y]\|^p\|z_t^m[y]-z_t^m[\tilde y]\|^p\right]^{1/p}\notag\\
&+\!\sup_{t\in[0,\tilde t]}\!\!E\left[\|z_t^m[ y]-z_t^m[\tilde y]\|^p\|z_t^m[\tilde y]\|^p\right]^{1/p}\!+\!\!\sup_{t\in[0,\tilde t]}\!\!E\left[\|y_t-\tilde y_t\|^p\|z_t^m[\tilde y]\|^{2p}\right]^{1/p}\!\bigg).\notag
\end{align}
Here and in the following, $\|x^i\|$ denotes the $\ell^2$ norm of the vector with components $x^i$. The other, paired, indices still indicate summations.

Let $p<\tilde q<q$.  Using H\"older's inequality (H)  gives
\begin{align}
&\sup_{t\in[0,\tilde t]}E\left[\|\sqrt{m}(Q^{ikl}G_{kl}^{ab})(t,y_t)(z_t^m[y])_a(z_t^m[y])_b\right.\\
&\left.-\sqrt{m}(Q^{ikl}G_{kl}^{ab})(t,\tilde y_t)(z_t^m[\tilde y])_a(z_t^m[\tilde y])_b\|^p\right]^{1/p}\notag\\
\leq&L\bigg(\sup_{t\in[0,\tilde t]}E\left[\|z_t^m[y]\|^{p\tilde q/(\tilde q-p)}\right]^{(\tilde q-p)/(p\tilde q)}\sup_{t\in[0,\tilde t]}E\left[\|z_t^m[y]-z_t^m[\tilde y]\|^{\tilde q}\right]^{1/\tilde q}&\text{(H)}\notag\\
&+\sup_{t\in[0,\tilde t]}E\left[\|z_t^m[\tilde y]\|^{p\tilde q/(\tilde q-p)}\right]^{(\tilde q-p)/(p\tilde q)}\sup_{t\in[0,\tilde t]}E\left[\|z_t^m[y]-z_t^m[\tilde y]\|^{\tilde q}\right]^{1/\tilde q}&\text{(H)}\notag\\
&+\sup_{t\in[0,\tilde t]}E\left[\|y_t-\tilde y_t\|^q\right]^{1/q}\sup_{t\in[0,\tilde t]}E\left[\|z_t^m[\tilde y]\|^{2pq/(q-p)}\right]^{(q-p)/(qp)}\bigg)&\text{(H)}\notag\\
\leq&L\sup_{t\in[0,\tilde t]}E\left[\|y_t-\tilde y_t\|^q\right]^{1/q}.\notag
\end{align}
To obtain the last line, we used Lemma \ref{z_lip_lemma}.

The third, fourth and fifth terms in \req{eq:R_grouping} are bounded similarly, using also the facts that $F$ and $\partial_s(\tilde\gamma^{-1})^{ij}$, $\partial_{q^c}(Q^{ikl}G_{kl}^{ab})$, $\partial_{s}(Q^{ikl}G_{kl}^{ab})$ are bounded and Lipschitz in $q$ (as are implied by the assumptions in Appendix \ref{app:assump}). For example, defining 
\begin{align}
H^j(t,q)\equiv  \left(\partial_s(\tilde\gamma^{-1})^{ij}+Q^{ikl}G_{kl}^{jb}F_b+Q^{ikl}G_{kl}^{aj} F_a\right)(t,q),
\end{align}
the estimate of the third term is derived by first using Minkowski's inequality for integrals (MI) to write
\begin{align}
&\sup_{t\in[0,\tilde t]}E\left[\|\int_0^t(z_s^m[y])_jH^j(s,y_s)ds-\int_0^t(z_s^m[\tilde y])_jH^j(s,\tilde y_s)ds\|^p\right]^{1/p}\\
\leq&\int_0^{\tilde t}E\left[\|(z_s^m[y])_jH^j(s,y_s)-(z_s^m[\tilde y])_jH^j(s,\tilde y_s)\|^p\right]^{1/p}ds.\notag
\end{align}
The rest of the derivation mimics that of the first two terms.

Finally, for the last term define 
\begin{align}
L^j(t,q)=\left(Q^{ikl}G_{kl}^{jb}\sigma_{b\rho}+Q^{ikl}G_{kl}^{aj}\sigma_{a\rho}\right)(t,q).
\end{align}
This is bounded and Lipschitz, and the Burkholder-Davis-Gundy inequality (BDG) followed by Minkowski's inequality for integrals (MI)  give
\begin{align}
&\sup_{t\in[0,\tilde t]}E\left[\|\int_0^t (z_s^m[y])_j L^j(s,y_s)dW^\rho_s-\int_0^t (z_s^m[\tilde y])_j L^j(s,\tilde y_s)dW^\rho_s\|^p\right]^{1/p}\\
\leq&L \left(\int_0^{\tilde t}E\left[ \|(z_s^m[y])_j L^j(s,y_s)- (z_s^m[\tilde y])_j L^j(s,\tilde y_s)\|^p\right]^{2/p}ds\right)^{1/2}.\notag
\end{align}
The rest of the proof proceeds similarly to the previous cases.
 
This completes the proof of the Lipschitz property. The proof of the inequality
\begin{align}
\sup_{t\in[0,T]}E\left[\|R_t^m[y]\|^p\right]\leq C
\end{align}
with $C$ independent of $m$ follows from \req{sup_E_z_bound} using similar techniques.  We omit the details.
\end{proof}

We end this section with Lipschitz and boundedness properties for $R_t^m$, corresponding to the norms from Lemma \ref{lemma:E_sup_z}.
\begin{lemma}\label{lemma:E_sup_R_bound}
Let $y,\tilde y\in Y$.  Then for any $m_0>0$,  $T>0$, $p\geq 2$, $\epsilon>0$ there exist $q>p$, $C$, $L$ such that for $0<m\leq m_0$, $0\leq t\leq T$:
\begin{align}\label{E_sup_R_lipschitz_bound}
&E\left[\sup_{s\in [0,t]}\|R_s^m[y]-R_s^m[\tilde y]\|^p\right]^{1/p}\leq \frac{L}{m^\epsilon}E\left[\sup_{s\in [0,t]}\|y_s-\tilde y_s\|^{q}\right]^{1/{q}}
\end{align}
and
\begin{align}\label{E_sup_R_bound}
E\left[\sup_{t\in [0,T]}\|R_t^m[y]\|^p\right]^{1/p}\leq C/m^\epsilon.
\end{align}
 $C$ and $L$ are independent of $m$.

\end{lemma}
\begin{proof}
The proof is very similar to that of Lemma \ref{lemma:sup_E_R_bound}, with the bounds and Lipschitz constants for $z_t^m[y]$ coming from  Lemma \ref{lemma:E_sup_z}. We omit the details.
\end{proof}

\subsection{Hierarchy of Approximations and the Convergence Proof}\label{sec:conv_proof}
With the notation of the previous subsection, $q_t^m$ solves the delay equation
\begin{align}\label{q_delay_eq}
dq_t^m=&\tilde\gamma^{-1}(t,q_t^m)F(t,q^m_t)dt+S(t,q_t^m)dt\\
&+\tilde\gamma^{-1}(t,q_t^m)\sigma(t,q_t^m)dW_t+\sqrt{m}dR^m_t[q^m].\notag
\end{align}
We use this form of the equation to motivate the definition of a hierarchy of approximating processes, $q_t^{\ell,m}$, and prove the claimed convergence result, \req{conv_rate}.  We first recall the following definitions for convenience:
\begin{definition}\label{main_def}
For $y$ a continuous martingale, define 
\begin{align}\label{R_main_def}
d(R^m_t[y])^i&\equiv- d((\tilde\gamma^{-1})^{ij}(t,y_t)(z^m_t[y])_j)+(z_t^m[y])_j\partial_t(\tilde\gamma^{-1})^{ij}(t,y_t)dt\\
&+Q^{ikl}(t,y_t)G_{kl}^{ab}(t,y_t)\left((z_t^m[y])_a F_b(t,y_t)+(z_t^m[y])_b F_a(t,y_t)\right)dt\notag\\
& +(z_t^m[y])_a(z_t^m[y])_b(z_t^m[y])^c\partial_{q^c}(Q^{ikl}G_{kl}^{ab})(t,y_t) dt\notag\\
&+Q^{ikl}(t,y_t)G_{kl}^{ab}(t,y_t)\left((z_t^m[y])_a\sigma_{b\rho}(t,y_t)+(z_t^m[y])_b\sigma_{a\rho}(t,y_t)\right)dW^\rho_t\notag\\
&-\sqrt{m}d(Q^{ikl}(t,y_t)G_{kl}^{ab}(t,y_t)(z_t^m[y])_a(z_t^m[y])_b)\notag\\
&+ \sqrt{m} (z_t^m[y])_a(z_t^m[y])_b \partial_t(Q^{ikl}G_{kl}^{ab})(t,y_t)dt,\notag
\end{align}
$R_0^m[y]\equiv 0$, where
\begin{align}\label{z_func_def2}
z^m_t[y]\equiv&\Phi^m_t[y]z^m_0+\frac{1}{\sqrt{m}}\Phi^m_t[y]\int_0^t (\Phi^m_s[y])^{-1}F(s,y_s) ds\\
&+\frac{1}{\sqrt{m}}\Phi^m_t[y]\int_0^t(\Phi^m_s[y])^{-1} \sigma(s,y_s) dW_s,\notag
\end{align}
 $\Phi_t^m[y]$ is defined pathwise as the solution to
\begin{align}\label{Phi_def}
\frac{d}{dt}\Phi_t^m=-\frac{1}{m}\tilde \gamma(t,y_t)\Phi^m_t,\hspace{2mm} \Phi^m_0=I,
\end{align}
and
\begin{enumerate}
\item $\tilde \gamma_{ik}(t,q)\equiv\gamma_{ik}(t,q) +\partial_{q^k}\psi_i(t,q)-\partial_{q^i}\psi_k(t,q)$,
\item $Q^{ijl}(t,q)\equiv \partial_{q^k}(\tilde\gamma^{-1})^{ij}(t,q) \delta^{kl}$,
\item $S^i(t,q)\equiv  \partial_{q^k}(\tilde\gamma^{-1})^{ij}(t,q) \delta^{kl}G_{jl}^{rs}(t,q)\Sigma_{rs}(t,q)$,
\item $G_{ij}^{kl}(t,q)\equiv \delta^{rk}\delta^{sl}\int_0^\infty (e^{-\zeta \tilde\gamma(t,q)})_{ir} (e^{-\zeta \tilde\gamma(t,q)})_{js} d\zeta$,
\item $\Sigma_{ij}\equiv  \sigma_{i\rho}\sigma_{j\xi}\delta^{\rho\xi}$,
\item $F(t,q)=-\partial_t\psi(t,q)-\nabla_q V(t,q)+\tilde F(t,q)$.
\end{enumerate}
\end{definition}

\begin{theorem}\label{convergence_theorem}
Assume the conditions in Appendix \ref{app:assump} hold. Fix an initial condition $q_0$ such that $E[\|q_0\|^p]<\infty$ for all $p>0$ and let $q_t^m$, $q_t$ be the solutions to the original SDE (\req{q_eq}) and the homogenized SDE (\req{q_SDE}) respectively, all with the same initial position, $q_0$.

With the notation from Definition \ref{main_def}, define the continuous semimartingales $q_t^{\ell,m}$, $\ell\in\mathbb{Z}^+$, by setting $q_t^{1,m}\equiv q_t$ and, for $\ell>1$, inductively defining $q_t^{\ell,m}$ to be the solution to
\begin{align}\label{q_ell_def}
q_t^{\ell,m}=&q_0+\int_0^t\tilde \gamma^{-1}(s,q^{\ell,m}_s)F(s,q^{\ell,m}_s)ds+\int_0^tS(s,q^{\ell,m}_s)ds\\
&+\int_0^t\tilde \gamma^{-1}(s,q^{\ell,m}_s)\sigma(s,q^{\ell,m}_s) dW_s +\sqrt{m}R^{m}_t[q^{\ell-1,m}].\notag
\end{align}
(Note that this also holds for $\ell=1$ if one defines $R^m_t[q^{0,m}]\equiv 0$.)

Then for  any $T>0$, $p>0$, $\epsilon>0$, $\ell\in\mathbb{Z}^+$ we have
\begin{align}\label{hierarchy_conv_rate}
&\sup_{t\in[0,T]}E[\|q^m_t-q^{\ell,m}_t\|^p]^{1/p}=O(m^{\ell/2}),\\
&E\left[\sup_{t\in[0,T]}\|q^m_t-q^{\ell,m}_t\|^p\right]^{1/p}=O(m^{\ell/2-\epsilon}).\notag
\end{align}

\end{theorem}
\begin{proof}
Given $y\in Y$ (recall \req{Y_def}), Theorem \ref{gen_sde_existence} and Corollary \ref{lin_growth_corollary} give a unique solution to
\begin{align}
x_t=&q_0+\int_0^t\tilde \gamma^{-1}(s,x_s)F(s,x_s)ds+\int_0^tS(s,x_s)ds\\
&+\int_0^t\tilde \gamma^{-1}(s,x_s)\sigma(s,x_s) dW_s +\sqrt{m}R^{m}_t[y]\notag
\end{align}
defined for all $t\geq 0$. The solution is a continuous semimartingale and boundedness of the drift and diffusion, together with \req{E_sup_R_bound}, imply 
\begin{align}
E\left[\sup_{t\in [0,T]}\|x_t\|^p\right]<\infty \text{ for all $T>0$, $p>0$},
\end{align}
i.e., $x\in Y$ as well.

Recalling \req{q_Lp_bound}, we see that  $q_t^{1,m}=q_t\in Y$. Therefore, the inductive definition of the $q_t^{\ell,m}$ results in a sequence of processes in $Y$.

To prove \req{hierarchy_conv_rate} by induction, first note that  the $\ell=0$ result is the content of \req{results_summary1a}, proven in \cite{BirrellHomogenization}.  Supposing that \req{hierarchy_conv_rate} holds for $\ell\in\mathbb{Z}^+$ and any $T>0$, $p>0$, $\epsilon>0$, we now show it holds for $\ell+1$:

 \req{q_Lp_bound} implies $q_t^m\in Y$ for all $m>0$ as well. From \req{q_delay_eq}  and \req{q_ell_def}, we have
\begin{align}\label{delta_q_formula}
&q^m_t-q^{\ell+1,m}_t=\int_0^t(\tilde \gamma^{-1}F+S)(s,q^{m}_s)-(\tilde \gamma^{-1}F+S)(s,q^{\ell+1,m}_s)ds\\
&+\int_0^t(\tilde \gamma^{-1}\sigma)(s,q^{m}_s)-(\tilde \gamma^{-1}\sigma)(s,q^{\ell+1,m}_s) dW_s +\sqrt{m}\left(R^{m}_t[q^{m}]-R^{m}_t[q^{\ell,m}]\right).\notag
\end{align}

Note that $(\tilde\gamma^{-1}F+S)(t,q)$ is Lipschitz in $q$, uniformly in $t\in[0,T]$. Let L denote a constant that may vary from line to line and   $0\leq \tilde t\leq T$, $q>p\geq 2$. The inequalities in Appendix \ref{app:ineq}  give
\begin{align}
&\sup_{t\in[0,\tilde t]}E[\|q^m_t-q^{\ell+1,m}_t\|^p]\\
\leq&3^{p-1}\bigg(L^p\sup_{t\in[0,\tilde t]}E\left[ \left(\int_0^t\|q^{m}_s-q^{\ell+1,m}_s\|ds\right)^p\right]&\text{(HFS)}\notag\\
&+\sup_{t\in[0,\tilde t]}E\left[\|\int_0^t(\tilde \gamma^{-1}\sigma)(s,q^{m}_s)-(\tilde \gamma^{-1}\sigma)(s,q^{\ell+1,m}_s) dW_s\|^p\right]&\text{(HFS)}\notag\\
& +m^{p/2}\sup_{t\in[0,\tilde t]}E\left[\|R^{m}_t[q^{m}]-R^{m}_t[q^{\ell,m}]\|^p\right]\bigg)&\text{(HFS)}\notag\\
\leq&3^{p-1}\bigg(L^p\sup_{t\in[0,\tilde t]}E\left[ \left(\int_0^t\|q^{m}_s-q^{\ell+1,m}_s\|ds\right)^p\right]\notag\\
&+\sup_{t\in[0,\tilde t]}E\left[\left(\int_0^t\|(\tilde \gamma^{-1}\sigma)(s,q^{m}_s)-(\tilde \gamma^{-1}\sigma)(s,q^{\ell+1,m}_s)\|^2 ds\right)^{p/2}\right]&\text{(BDG)}\notag\\
& +m^{p/2}L^p\sup_{t\in[0,\tilde t]}E\left[\|q^{m}_t-q_t^{\ell,m}\|^q\right]^{p/q}\bigg).\notag
\end{align}
Lemma   \ref{lemma:sup_E_R_bound} was used to obtain the last line.

$\tilde \gamma^{-1}\sigma$ is Lipschitz in $q$, uniformly in $t\in[0,T]$. This, together with  H\"older's inequality and the induction hypothesis yields
\begin{align}
&\sup_{t\in[0,\tilde t]}E[\|q^m_t-q^{\ell+1,m}_t\|^p]\\
\leq&3^{p-1}\bigg(T^{p-1}L^p \int_0^{\tilde t}E\left[\|q^{m}_s-q^{\ell+1,m}_s\|^p\right]ds&\text{(H)}\notag\\
&+T^{p/2-1}L^p\int_0^{\tilde t}E\left[\|q^{m}_s-q^{\ell+1,m}_s\|^p \right]ds +O(m^{p(\ell+1)/2})\bigg)&\text{(H)}\notag\\
\leq&C\int_0^{\tilde t}\sup_{r\in[0,s]}E\left[\|q^{m}_r-q^{\ell+1,m}_r\|^p \right]ds+O(m^{p(\ell+1)/2})\notag
\end{align}
for some constant $C>0$.

The integrand is $L^1$ because $q_t^m$, $q^{\ell+1,m}_t\in Y$, hence Gronwall's inequality gives
\begin{align}
\sup_{t\in[0,T]}E[\|q^m_t-q^{\ell+1,m}_t\|^p]\leq O(m^{p(\ell+1)/2})e^{CT}=O(m^{p(\ell+1)/2}).
\end{align}
We have proven the desired bound for $p\geq 2$, but it follows for all $p>0$ by H\"older's inequality (H).  This completes the proof of the first part of \req{hierarchy_conv_rate}.

The proof of the second part is similar.  Again starting from \req{delta_q_formula}, for $p\geq 2$   we have
\begin{align}
&E\left[\sup_{t\in[0,\tilde t]}\|q^m_t-q^{\ell+1,m}_t\|^p\right]\\
\leq& 3^{p-1} \bigg(L^pE\left[\left(\int_0^{\tilde t}\|q^{m}_s-q^{\ell+1,m}_s\|ds\right)^p\right]&\text{(HFS)}\notag\\
&+E\left[\sup_{t\in[0,\tilde t]}\|\int_0^t(\tilde \gamma^{-1}\sigma)(s,q^{m}_s)-(\tilde \gamma^{-1}\sigma)(s,q^{\ell+1,m}_s) dW_s\|^p\right]&\text{(HFS)}\notag\\
& +m^{p/2}E\left[\sup_{t\in[0,\tilde t]}\|R^{m}_t[q^{m}]-R^{m}_t[q^{\ell,m}]\|^p\right]\bigg)&\text{(HFS)}\notag\\
\leq&3^{p-1} \bigg(L^pT^{p-1}\int_0^{\tilde t}E\left[\sup_{r\in[0,s]}\|q^{m}_r-q^{\ell+1,m}_r\|^p\right]ds&\text{(H)}\notag\\
&+T^{p/2-1}L^p\int_0^{\tilde t}E\left[\sup_{r\in[0,s]}\|q^{m}_r-q^{\ell+1,m}_r\|^p \right]ds&\text{(BDG,H)}\notag\\
& +m^{p/2}E\left[\sup_{t\in[0,\tilde t]}\|R^{m}_t[q^{m}]-R^{m}_t[q^{\ell,m}]\|^p\right]\bigg).\notag
\end{align}
By \req{E_sup_R_lipschitz_bound} and the induction hypothesis, given $\epsilon>0$ there exist $L>0$, $q>0$ such that
\begin{align}
&E\left[\sup_{t\in[0,\tilde t]}\|R^{m}_t[q^{m}]-R^{m}_t[q^{\ell,m}]\|^p\right]\leq\frac{L^p}{m^{p\epsilon/2}}E\left[\sup_{t\in[0,\tilde t]}\|q^{m}-q^{\ell,m}\|^q\right]^{p/q}\\
\leq&\frac{L^p}{m^{p\epsilon/2}}O(m^{p(\ell-\epsilon)/2})=O(m^{p(\ell/2-\epsilon)}).\notag
\end{align}
Therefore
\begin{align}
&E\left[\sup_{t\in[0,\tilde t]}\|q^m_t-q^{\ell+1,m}_t\|^p\right]\\
\leq&C\int_0^{\tilde t}E\left[\sup_{r\in[0,s]}\|q^{m}_r-q^{\ell+1,m}_r\|^p\right]ds+O(m^{p((\ell+1)/2-\epsilon)})\notag
\end{align}
for some $C>0$.

The integrand is again in $L^1$ because $q_t^m$, $q^{\ell+1,m}_t\in Y$, hence Gronwall's inequality similarly gives
\begin{align}
&E\left[\sup_{t\in[0,T]}\|q^m_t-q^{\ell+1,m}_t\|^p\right]\leq O(m^{p((\ell+1)/2-\epsilon)}).
\end{align}
The bound for arbitrary $p>0$ again follows from H\"older's inequality (H), and so the proof of the second half of \req{hierarchy_conv_rate} is also complete.

\end{proof}

\begin{remark}
By introducing auxiliary variables  $z^{\ell-1,m}_t\equiv z_t^m[q^{\ell-1,m}]$, noting that they satisfy 
\begin{align}\label{z_SDE}
dz^{\ell-1,m}_t=& -\frac{1}{m}\tilde\gamma(s,q_s^{\ell-1,m})z_t^{\ell-1,m}dt+\frac{1}{\sqrt{m}}F(t,q_t^{\ell-1,m})dt\\
&+\frac{1}{\sqrt{m}}\sigma(t,q_t^{\ell-1,m})dW_t,\notag
\end{align}
and using It\^o's formula on the terms
\begin{align}
- d((\tilde\gamma^{-1})^{ij}(t,y_t)(z^m_t[y])_j) \text{ and } 
-\sqrt{m}d(Q^{ikl}(t,y_t)G_{kl}^{ab}(t,y_t)(z_t^m[y])_a(z_t^m[y])_b)\notag
\end{align}
in $R^m_t[q^{\ell-1,m}]$, the hierarchy \req{q_ell_def} can be embedded in a system of SDEs.  However, for our purposes the resulting form is much less convenient to work with than the hierarchy \req{q_ell_def}, largely due to singular nature of the $1/m$ and $1/\sqrt{m}$ factors in \req{z_SDE}.  In contrast, the $m$-dependence of the integral formula for $z_t^m[q^{\ell-1,m}]$ (see \req{z_func_def2} and \req{Phi_def},  especially once combined with \req{fund_sol_decay} and \req{convol_decomp})  is manifestly  more well-behaved.  Similarly, the $m$-dependence of the formula for $R^m_t[q^{\ell-1,m}]$ (see \req{R_main_def}) and of the equations for  $q^{\ell,m}_t$ (see \req{q_ell_def}) presents no additional trouble.  These facts  play a crucial role in our proofs.
\end{remark}

\subsection{Special Cases}\label{sec:special_cases}
We end this section by presenting simplified formulas for the hierarchy  in several important special cases, all of which are direct consequences of the theorems stated earlier. First we recall the noise-induced drift in the fluctuation-dissipation case:
\begin{corollary}
Suppose that $\psi=0$ and a fluctuation-dissipation relation holds,
\begin{align}
\Sigma_{ij}(t,q)=2k_BT(t,q)\gamma_{ij}(t,q),
\end{align}
for a time and position dependent `temperature' $T(t,q)$.  Then the noise-induced drift has the following simplified form:
\begin{align}\label{S_simp}
S^i(t,q)=  k_BT(t,q)\partial_{q^j}(\gamma^{-1})^{ij}(t,q).
\end{align}
\end{corollary}
While \req{S_simp} greatly simplifies the first approximation, the full hierarchy is still quite complicated.  Things simplify further in the scalar case:
\begin{corollary}
Suppose that $\psi=0$ and $\gamma$ and $\sigma$ are scalar-valued. Note that this automatically gives a fluctuation dissipation relation with
\begin{align}
T(t,q)=\frac{\sigma^2(t,q)}{2k_B\gamma(t,q)}.
\end{align}

Under these conditions, the approximating hierarchy, \req{q_ell_def}, takes the following  form for $\ell>1$:
\begin{align}\label{q_ell_def_simp}
q_t^{\ell,m}=&q_0+\int_0^t \gamma^{-1}(s,q^{\ell,m}_s)F(s,q^{\ell,m}_s)ds+\int_0^t  k_BT(s,q^{\ell,m}_s)\nabla_q(\gamma^{-1})(s,q^{\ell,m}_s)ds\\
&+\int_0^t \gamma^{-1}(s,q^{\ell,m}_s)\sigma(s,q^{\ell,m}_s) dW_s +\sqrt{m}R^{m}_t[q^{\ell-1,m}],\notag
\end{align}
where
\begin{align}
d(R^m_t[y])^i=&- d((\gamma^{-1})^{ij}(t,y_t)(z^m_t[y])_j)+(z_t^m[y])_j\partial_t(\gamma^{-1})^{ij}(t,y_t)dt\\
&+Y^{ikl}(t,y_t)\left((z_t^m[y])_k F_l(t,y_t)+(z_t^m[y])_l F_k(t,y_t)\right)dt\notag\\
& +(z_t^m[y])_k(z_t^m[y])_l(z_t^m[y])^j\partial_{q^j}(Y^{ikl})(t,y_t) dt\notag\\
&+Y^{ikl}(t,y_t)\left((z_t^m[y])_k\sigma_{l\rho}(t,y_t)+(z_t^m[y])_l\sigma_{k\rho}(t,y_t)\right)dW^\rho_t\notag\\
&-\sqrt{m}d(Y^{ikl}(t,y_t)(z_t^m[y])_k(z_t^m[y])_l)\notag\\
&+ \sqrt{m} (z_t^m[y])_k(z_t^m[y])_l \partial_t(Y^{ikl})(t,y_t)dt,\notag
\end{align}
$R_0^m[y]= 0$,  
\begin{align}
&z^m_t[y]=\exp\left(-\frac{1}{m} \int_0^t\gamma(r,y_r)dr\right)z^m_0\\
&+\frac{1}{\sqrt{m}}\int_0^t\exp\left(-\frac{1}{m}\int_s^t\gamma(r,y_r)dr\right)F(s,y_s) ds\notag\\
&+\frac{1}{\sqrt{m}}\exp\left(-\frac{1}{m} \int_0^t\gamma(r,y_r)dr\right)\int_0^t\exp\left(\frac{1}{m} \int_0^s\gamma(r,y_r)dr\right) \sigma(s,y_s) dW_s,\notag
\end{align}
\begin{align}
Y^{ikl}(t,q)\equiv&  \frac{1}{2}\gamma^{-1}(t,q) \partial_{q^j}\gamma^{-1}(t,q)\delta^{ik} \delta^{jl},\\
F(t,q)=&-\nabla_q V(t,q)+\tilde F(t,q).
\end{align}

\end{corollary}

Finally, instead of a fluctuation-dissipation relation, suppose that $\gamma$ is state-independent:
\begin{corollary}
Suppose that $\psi=0$ and $\gamma$ is independent of $q$.  Then  the approximating hierarchy, \req{q_ell_def}, takes the following  form for $\ell>1$:
\begin{align}
q_t^{\ell,m}=&q_0+\int_0^t \gamma^{-1}(s)F(s,q^{\ell,m}_s)ds+\int_0^t \gamma^{-1}(s)\sigma(s,q^{\ell,m}_s) dW_s\\
&+\sqrt{m}R^{m}_t[q^{\ell-1,m}]\notag,
\end{align}
where $F(t,q)=-\nabla_q V(t,q)+\tilde F(t,q)$,
\begin{align}
&R^{m}_t[q^{\ell-1,m}]\\
=&\int_0^t \partial_t(\gamma^{-1})(s)z_s^m[q^{\ell-1,m}]ds -\left( \gamma^{-1}(t)z^m_t[q^{\ell-1,m}]-\gamma^{-1}(0)z^m_0\right),\notag
\end{align}
\begin{align}
&z^m_t[y]\\
=&\Phi^m_tz^m_0+\frac{1}{\sqrt{m}}\Phi^m_t\left(\int_0^t (\Phi^m_s)^{-1}F(s,y_s) ds+\int_0^t(\Phi^m_s)^{-1} \sigma(s,y_s) dW_s\right),\notag
\end{align}
and  $\Phi_t^m$ is the (non-random) matrix-valued function that solves
\begin{align}
\frac{d}{dt}\Phi_t^m=-\frac{1}{m} \gamma(t)\Phi^m_t,\hspace{2mm} \Phi^m_0=I.
\end{align}
\end{corollary}

\section{Extension to Unbounded Forcing}\label{sec:unbounded} 

The boundedness assumptions in  Theorem \ref{convergence_theorem} can be relaxed by using the  technique developed in \cite{herzog2015small}, and similarly used in \cite{BirrellHomogenization}, at the cost of weakening the convergence mode to convergence in probability.  Specifically, here we will focus on accommodating  unbounded forces, $\tilde F$ and $\nabla_qV$, where $V$ is sufficiently confining.  In this section, we will no longer be working under the assumptions from Appendix \ref{app:assump}, but rather:
\begin{assumption}\label{unbounded_assump}
Assume that:

\begin{enumerate}
\item $V(t,q)$ is $C^2$ and there exist $a\geq 0, b\geq 0$ such that 
\begin{align}\label{tilde_V_def}
\tilde V(t,q)\equiv a+b\|q\|^2+V(t,q)
\end{align}
 is non-negative.
\item  $\psi(t,q)$ is  $C^4$ and  $\nabla_q\psi$ is bounded.
\item $\gamma(t,q)$ is a bounded, $C^3$ function valued in the symmetric $n\times n$ real matrices with eigenvalues bounded below by some $\lambda>0$.
\item  $\sigma(t,q)$ is bounded, continuous and   Lipschitz in $q$ with the Lipschitz constant uniform on compact time intervals.
\item  $\tilde F(t,q)$ is continuous and  locally Lipschitz in $q$ with the Lipschitz constant uniform on compact time intervals.
\item There exist $C>0$, $M>0$ such that
\begin{align}
|\partial_t V(t,q)|\leq M+C(\|q\|^2+\tilde V(t,q)),
\end{align}
\begin{align}\label{eq:dt_psi_F_bound}
\|-\partial_t\psi(t,q)+\tilde F(t,q)\|^2\leq M+C\left(\|q\|^2+\tilde V(t,q)\right),
\end{align}
\begin{align}
\|\partial_{q^i}\tilde\gamma(t,q)\|^2\leq M+C\left(\|q\|^2+\tilde V(t,q)\right), \hspace{2mm} i=1,...,n,
\end{align}
\begin{align}\label{eq:nablaV_bound}
\|\nabla_qV(t,q)\|\leq M+C(\|q\|^2+\tilde V(t,q)),
\end{align}
and
\begin{align}
\left(\sum_{i,j}|\partial_{q^i}\partial_{q^j} V(t,q)|^2\right)^{1/2}\leq M+C\left(\|q\|^2+\tilde V(t,q)\right).
\end{align}
\item We have $\mathbb{R}^n$-valued initial conditions ($\mathcal{F}_0$-measurable random variables)  $(q_0,u_0^m)$ that satisfy the following:

$E[\|q_0\|^p]<\infty$ for all $p>0$ and there exists $C>0$ such that  $ \|u^m_0\|^2 \leq Cm$ for all $m>0$ and all $\omega\in\Omega$.
\end{enumerate}
\end{assumption}
  These assumptions are not of utmost generality, but they are still quite general, are commonly satisfied, and are relatively convenient to work with.  These assumptions are similar, but not identical, to those of  Theorem 6.1 in  \cite{BirrellHomogenization} (recall that what we call $\tilde F$ here was called $F$ in \cite{BirrellHomogenization}); we will comment on the differences below. Under the above assumptions, we are able to prove the following theorem.
\begin{theorem}\label{thm:conv_in_prob}
Under Assumption \ref{unbounded_assump}, let $(q_t^m,u_t^m)$ be the solutions to \req{q_eq}-\req{u_eq} and $q_t$ to \req{q_SDE}. Define  the continuous semimartingales $q_t^{\ell,m}$, $\ell\in\mathbb{Z}^+$, by setting $q_t^{1,m}\equiv q_t$ and inductively defining $q_t^{\ell,m}$ to be the unique maximal solution to
\begin{align}\label{q_ell_gen_def}
q_t^{\ell,m}=&q_0+\int_0^t\tilde \gamma^{-1}(s,q^{\ell,m}_s)F(s,q^{\ell,m}_s)ds+\int_0^tS(s,q^{\ell,m}_s)ds\\
&+\int_0^t\tilde \gamma^{-1}(s,q^{\ell,m}_s)\sigma(s,q^{\ell,m}_s) dW_s +\sqrt{m}R^{m}_t[q^{\ell-1,m}].\notag
\end{align} 

Then all $q_t^m$, $q_t^{\ell,m}$ are continuous semimartingales, they are defined for all $t\geq 0$, and
 \begin{align}\label{eq:unbounded_conv}
\lim_{m\to 0} P\left(\frac{\sup_{t\in[0,T]}\|q^m_t-q^{\ell,m}_t\|}{m^{\ell/2-\epsilon}}>\delta\right)=0
\end{align}
for all $T>0$, $\delta>0$, $\epsilon>0$, $\ell\in\mathbb{Z}^+$.
 \end{theorem}
The method of proof is very similar to that of Theorem 6.1 in  \cite{BirrellHomogenization}, though there are some additional technical complications, primarily arising from the need to prove non-explosion of solutions to various SDEs with unbounded coefficients and semimartingale external-forcing terms (see Appendix \ref{app:gen_SDE} for the relevant tools). The main ideas are outlined below and, for completeness, a full proof is included in Appendix \ref{app:unbounded_F_proof}:

The majority of the items in Assumption \ref{unbounded_assump} are there to ensure non-explosion of solutions to the Langevin equation, as well as the approximation hierarchy, i.e., so that we can apply the Liapunov function method from Theorem \ref{gen_liap_thm} (see the assumptions therein).  These requirements differ somewhat from those in \cite{BirrellHomogenization}, as   here we must consider SDEs with external semimartingale forcing terms.  Our method for handling such systems  requires the additional bound \req{eq:nablaV_bound} (compare with item (4) in Theorem  \ref{gen_liap_thm}).  On the other hand, we combined  Eq. (6.8)-(6.9) from Theorem 6.1 of \cite{BirrellHomogenization} into the single condition \req{eq:dt_psi_F_bound}  simply for efficiency, as it does not change the proof.

To prove the convergence result \req{eq:unbounded_conv}, one begins by defining a family of cutoff systems: Let $\chi:\mathbb{R}^n\to [0,1]$ be a $C^\infty$ bump function, equal to $1$ on  $\overline{B_1(0)}\equiv\{\|q\|\leq 1\}$ and zero outside $\overline{B_2(0)}$. Given $r>0$ let $\chi_r(q)=\chi(q/r)$ and define
\begin{align}
&V_r(t,q)=\chi_r(q)V(t,q), \hspace{2mm} \tilde F_r(t,q)=\chi_r(q) \tilde F(t,q), \hspace{2mm}  \psi_r(t,q)=\chi_r(q)\psi(t,q),\notag\\
&\gamma_r(t,q)=\chi_r(q)\gamma(t,q)+(1-\chi_r(q))\lambda I.
\end{align}
For each $r>0$, replacing $V$ with $V_r$, $F$ with $F_r$ etc., we arrive at an SDE satisfying the hypotheses of Theorem \ref{convergence_theorem}; the regularity conditions in \req{unbounded_assump} were chosen  for this exact purpose.  Note that this also accounts for the stronger assumptions here, as opposed to in Theorem 6.1 from  \cite{BirrellHomogenization}; see also the discussion at the end of Appendix \ref{app:assump}.

  Let $(q_t^{r,m},u_t^{r,m})$ be the solutions to the cutoff system, $q_t^r$ the solution to the corresponding homogenized equation, and $q_t^{r,\ell,m}$ the solutions to the corresponding approximating hierarchy, all using the same initial conditions as the system without the cutoff. Corresponding solutions to the original and cutoff systems  agree up until the first exit time, denoted $\sigma_r^{\ell,m}$,  of any position processes (Langevin, homogenized, or hierarchy up to level $\ell$) from the ball of radius $r$. Therefore
\begin{align}
&P\left(\frac{\sup_{t\in[0,T]}\|q_t^m-q^{\ell,m}_t\|}{m^{\ell/2-\epsilon}}>\delta\right)\\
\leq&P\left(\sigma^{\ell,m}_r> T,\frac{\sup_{t\in[0,T]}\|q_{t}^{r,m}-q^{r,\ell,m}_{t}\|}{m^{\ell/2-\epsilon}}>\delta\right)+P\left(\sigma_r^{\ell,m}\leq T\right).\notag
\end{align}
The first term involves the cutoff processes, to which Theorem \ref{convergence_theorem} applies, and hence it converges to zero by Markov's inequality.  For the second, one must show that it converges to zero as $r\to\infty$, uniformly in $m$, i.e., the probability that any of the position processes exits the $r$-ball goes to zero as $r\to\infty$, uniformly in $m$.  This is done in Appendix \ref{app:unbounded_F_proof}, starting with \req{eq:r_ball_prob}.

\section{Discussion}

In this paper we have shown how a bootstrapping method can be used to derive higher-order approximations to the position degrees-of-freedom of Langevin dynamics, in the small-mass limit.  We obtain a hierarchy of approximations $q^{\ell,m}_t$ ($m$ denotes the mass), where the $\ell$'th level is an  $O(m^{\ell/2})$ approximation to the Langevin position degrees-of-freedom, $q^m_t$,  over compact time intervals.  The equations for the $q^{\ell,m}_t$'s (see Theorem \ref{convergence_theorem}) consist of the standard overdamped equation (i.e., the $\ell=1$ equation) with an added semimartingale correction term, which is independent of $q^{\ell,m}_t$; the correction term (for $\ell\geq 2$) is constructed from the solution, $q^{\ell-1,m}_t$, at the previous level.    This work naturally leads to the following two questions.

First, can the hierarchy of approximations derived here be used as the basis for efficient numerical methods with higher-order-in-$m$ accuracy? This is a question for future work, but the form of hierarchy derived here suggests that one should be able to avoid the difficulty inherent in the $O(m^{-1/2})$ divergence of the velocity degrees of freedom in the underdamped Langevin equation.  The SDE for $q^{\ell,m}_t$, \req{q_ell_def}, consists of the standard overdamped approximation with an explicit semimartingale correction term.  Methods for simulating the overdamped Langevin equation are well studied and much used,  so the question is whether methods can be devised to efficiently incorporate this correction term.  The fact that the correction term does not depend on the variables, $q^{\ell,m}_t$, that one is  solving for is promising.

Secondly, can the  method employed here be adapted to study the singular limit of other  SDEs, moving beyond the small-mass limit of Langevin dynamics?  We do not have an answer at this time, but we note that   \cite{10.1007/978-981-15-0294-1_4} generalizes the technique of  \cite{BirrellHomogenization} to derive a homogenized SDE for a larger class of systems, with convergence in the same $L^p$-sense over compact time intervals, and with explicit  remainder terms; in fact, \cite{BirrellHomogenization} itself studies more general noisy, dissipative Hamiltonian systems than just Langevin dynamics.  Such convergence results for the Langevin equation were the starting point for the bootstrapping method used in this paper, and so \cite{10.1007/978-981-15-0294-1_4} might provide a starting point for  higher-order approximations to more general singular limits of SDEs.  We do anticipate that proving the required Lipschitz properties of the remainder terms will be more difficult; here we made heavy use of the  formula \req{z_func_def} for the fast degrees-of-freedom, $z^m_t$, but  a similar expression is not available in general.

\appendix

\numberwithin{assumption}{section}
\section{Assumptions Implying Homogenization as $m\to 0$}\label{app:assump}
\setcounter{assumption}{0}
    \renewcommand{\theassumption}{\Alph{section}\arabic{assumption}}

    \renewcommand{\thelemma}{\Alph{section}\arabic{lemma}}
    \renewcommand{\thetheorem}{\Alph{section}\arabic{theorem}}
    \renewcommand{\thecorollary}{\Alph{section}\arabic{corollary}}

In this appendix, we give a list of properties  that, as shown in \cite{BirrellHomogenization}, are sufficient to  guarantee that the solutions to the SDE \req{q_eq}-\req{u_eq} satisfy the properties \req{results_summary1a}, \req{results_summary1b}, and \req{q_Lp_bound} (note that what we call $\tilde F$ here was simply called $F$ in \cite{BirrellHomogenization}). Some of the  assumptions below are strengthened, as compared to \cite{BirrellHomogenization}, in order to meet the needs of the current paper; we remark on this further below.

We assume that
\begin{enumerate}
\item $\gamma:[0,\infty)\times\mathbb{R}^n\to\mathbb{R}^{n\times n}$ is $C^3$.
\begin{enumerate}
\item The values of $\gamma$  are symmetric matrices.
\item The  eigenvalues of $\gamma$ are  uniformly bounded below by some $\lambda>0$.
\item  $\gamma$ is bounded.
 \item For all $T>0$ and all   multi-indices $\alpha$  with $1\leq |\alpha|\leq 3$, $\partial_{q^\alpha}\gamma$ is bounded uniformly for $(t,q)\in [0,T]\times\mathbb{R}^{n}$.
\item For all $T>0$ and all   multi-indices $\alpha$  with $0\leq |\alpha|\leq 2$, $\partial_{q^\alpha}\partial_t\gamma$ is bounded uniformly for $(t,q)\in [0,T]\times\mathbb{R}^{n}$.
\end{enumerate}
\item $\psi:[0,\infty)\times\mathbb{R}^n\to\mathbb{R}^n$ is $C^4$.
\begin{enumerate}
\item For all $T>0$ and all   multi-indices $\alpha$ with $1\leq |\alpha|\leq 4$, $\partial_{q^\alpha}\psi$ is bounded uniformly for $(t,q)\in [0,T]\times\mathbb{R}^{n}$.
\item For all $T>0$ and all   multi-indices $\alpha$  with $0\leq |\alpha|\leq 3$, $\partial_{q^\alpha}\partial_t\psi$ is bounded uniformly for $(t,q)\in [0,T]\times\mathbb{R}^{n}$.
\end{enumerate}
\item $\tilde F:[0,\infty)\times\mathbb{R}^n\to\mathbb{R}^n$ is continuous.
\begin{enumerate}
\item   $\tilde F$ is bounded.
\item   $\tilde F$  is Lipschitz in $q$ uniformly in $t$.
\end{enumerate}
\item $\sigma:[0,\infty)\times\mathbb{R}^n\to\mathbb{R}^{n\times k}$ is continuous.
\begin{enumerate}
\item   $\sigma$ is  bounded.
\item $\sigma$ is Lipschitz in $q$ uniformly in $t$.
\end{enumerate}
\item $V:[0,\infty)\times\mathbb{R}^n\to\mathbb{R}^n$ is $C^2$.
\begin{enumerate}
\item  $\nabla_q V$ is Lipschitz in $q$ uniformly in $t$.
\item For all $T>0$, $\nabla_qV$ is bounded uniformly for $(t,q)\in [0,T]\times\mathbb{R}^{n}$.
\item There exist $a,b\geq 0$ such that $\tilde V(t,q)\equiv a+b\|q\|^2+V(t,q)$ is non-negative for all $t,q$.
\item There exist $M,C> 0$ such that 
\begin{align}
|\partial_tV(t,q)|\leq M+C(\|q\|^2+\tilde V(t,q))
\end{align}
and
\begin{align}
\|-\partial_t\psi(t,q)+\tilde F(t,q)\|^2\leq M+C(\|q\|^2+\tilde V(t,q))
\end{align}
 for all $t,q$.
\end{enumerate}
\item There exists $C>0$ such that the (random) initial conditions satisfy $ \|u^m_0\|^2 \leq Cm$ for all $m>0$ and all $\omega\in\Omega$ and $E[\|q^m_0\|^p]<\infty$, $E[\|q_0\|^p]<\infty$, and $E[\|q_0^m-q_0\|^p]^{1/p}=O(m^{1/2})$ for all $p>0$.
\end{enumerate}

The various global-in-time properties  are used to prove non-explosion of solutions, while the properties over compact time intervals are used to prove convergence to the homogenized SDE in  \cite{BirrellHomogenization}.  The reason we needed to strengthen certain regularity properties here, as compared to \cite{BirrellHomogenization}, is so we can prove the required Lipschitz properties of the remainder terms, \req{R_func_def}, on compact time intervals;  this is in contrast to \cite{BirrellHomogenization}, where one only had to show that the remainder terms converge to zero as $m\to 0$.  For example,  the third line of \req{R_main_def} includes a $\partial_{q^c}Q^{ikl}$ term, which in turn involves $\partial_{q^c}\partial_{q^b} \tilde\gamma^{-1}$.  To ensure this is Lipschitz in $q$, we have  assumed that  $\tilde\gamma$ is $C^3$ with third derivative  being bounded  on compact time intervals; more precisely, we have assumed this of both $\gamma$ and $\partial_j\psi$, as these are used to construct $\tilde\gamma$.  This is why we  require conditions on the third derivative of $\gamma$ and the fourth derivative of $\psi$, as opposed to \cite{BirrellHomogenization} where we only required conditions on derivatives up to order two and three respectively.  Similar remarks apply to the other objects.

\section{Properties of the Fundamental Solution}\label{app:fund_sol}
Our derivations will require the use of several properties of the fundamental solution of a linear ordinary differential equation (ODE).  Specifically, we need to consider the process obtained by pathwise solving the linear ODE
\begin{align}
\frac{d}{dt}\Phi_t^m=-\frac{1}{m}\tilde \gamma(t,y_t)\Phi^m_t,\hspace{2mm} \Phi^m_0=I,
\end{align}
where $y$ is a continuous semimartingale. The process  $\Phi^m_t$ is adapted and pathwise $C^1$; we will call it the fundamental-solution process, as each of its paths is the  fundamental solution to a linear ODE.

The symmetric part of $ \tilde \gamma$, denoted by $\gamma$,  is assumed to have eigenvalues bounded below by $\lambda>0$ (see Appendix \ref{app:assump}).  This implies the following  crucial bound
\begin{align}\label{fund_sol_decay}
\|\Phi^m_t(\Phi^m_s)^{-1}\|\leq e^{-\lambda(t-s)/m}\,\,\, \text{ for all }t\geq s
\end{align}
(see, for example, p.86 of \cite{teschl2012ordinary}). Note that while the left hand side is random, the upper bound is not. As we have stated it, this bound requires  the use of the the $\ell^2$ operator norm. Otherwise, there is an additional constant multiplying the exponential.

We will also need the following bound on the difference between the fundamental solutions corresponding to two  linear ODEs.  See the Appendix to \cite{BirrellPhaseSpace}.
\begin{lemma}\label{fund_matrix_diverg}
Let $B_i:[0,T]\rightarrow\mathbb{R}^{n\times n}$; $i=1,2$, be continuous and suppose their symmetric parts have eigenvalues bounded above by $\mu$, uniformly in $t$.  Consider the fundamental solutions, $\Phi_i(t)$, satisfying
\begin{align}
\frac{d}{dt}\Phi_i(t)=B_i(t)\Phi_i(t),\,\, \Phi_i(0)=I.
\end{align}
  Then for any $0\leq t\leq T$ we have the bound
\begin{align}
\|\Phi_1(t)-\Phi_2(t)\|\leq e^{\mu t} \int_0^t \|B_1(s)-B_2(s)\| ds  .
\end{align}
\end{lemma}

We will need the following lemma concerning  stochastic convolutions, adapted from Lemma 5.1 in \cite{particle_manifold_paper}:

\begin{lemma}\label{conv_lemma}
Let  $B_s$ be a continuous adapted  $\mathbb{R}^{n\times n}$-valued processes.  Let $\Phi(t)$ be the fundamental-solution process, pathwise satisfying
\begin{align}
\frac{d}{dt}\Phi(t)=B(t)\Phi(t),\,\, \Phi(0)=I.
\end{align}
Let $V_s$ be a  continuous adapted  $\mathbb{R}^{n\times k}$-valued processes. Then we have the $P$-a.s. equality
\begin{align}\label{convol_decomp}
&\Phi(t)\int_0^t\Phi^{-1}(s) V_s dW_s\\
=&\Phi(t)\int_0^t V_s dW_s-\Phi(t)\int_0^t \Phi^{-1} (s)B(s) \left(\int_s^t V_r dW_r\right) ds\text{ for all $t$.}\notag
\end{align}
\end{lemma}

The following lemma will assist us in bounding processes having the form of the last term in \req{convol_decomp}. The proof is very similar to that of Lemma 5.1 in  \cite{particle_manifold_paper}, but we provide it for completeness.
\begin{lemma}\label{Phi_int_bound}
Let $V_s$ be a continuous adapted  $\mathbb{R}^{n\times k}$-valued process and $\alpha>0$. 

Then for every $j\in\mathbb{Z}_0$ there exists $C_j>0$ such that for all $T>0$, $\delta>0$ we have the $P$-a.s. bound
\begin{align}
\sup_{t\in[0,T]}&\int_0^t (t-s)^je^{-\alpha(t-s)}\|\int_s^t V_r dW_r\| ds\\
\leq \frac{C_j}{\alpha^{j+1}}\bigg(&\max_{\ell=1,...,N}\sup_{\tau\in[(\ell-1)\delta,\min\{(\ell+1)\delta,T\}]}\|\int_{(\ell-1)\delta}^{\tau}V_rdW_r\|\notag\\
&+e^{-\alpha \delta/2}  \sup_{\tau\in[0,T]}\|\int_0^\tau V_r dW_r\|\bigg),\notag
\end{align}
where $N=\min\{\ell\in\mathbb{Z}:\ell\delta\geq T\}$.   We emphasize that $C_j$ depends only on $j$. 
\end{lemma}
\begin{proof}

Suppose $\delta<T$.  First split
\begin{align}
&\sup_{t\in[0,T]}\int_0^t (t-s)^j e^{-\alpha(t-s)}\|\int_s^t V_r dW_r\| ds\\
\leq &\sup_{t\in[0,\delta]}\int_0^t(t-s)^j e^{-\alpha(t-s)}\|\int_s^t V_r dW_r\| ds\notag\\
&+\sup_{t\in[\delta,T]}\int_0^t(t-s)^j e^{-\alpha(t-s)}\|\int_s^t V_r dW_r\| ds.\notag
\end{align}
The first term can be bounded as follows.
\begin{align}\label{first_term_bound}
&\sup_{t\in[0,\delta]}\int_0^t (t-s)^je^{-\alpha(t-s)}\|\int_s^t V_r dW_r\| ds\\
=&\sup_{t\in[0,\delta]}\int_0^t (t-s)^je^{-\alpha(t-s)}\|\int_0^t V_r dW_r-\int_0^s V_r dW_r\| ds\notag\\
\leq&\frac{2}{\alpha^{j+1}}\sup_{0\leq \tau\leq\delta}\|\int_0^\tau V_r dW_r\|\int_0^{\alpha \delta} u^je^{-u}du.\notag
\end{align}

In the second term we  split the integral to obtain
\begin{align}\label{second_term_bound1}
&\sup_{t\in[\delta,T]}\int_0^t(t-s)^j e^{-\alpha(t-s)}\|\int_s^t V_r dW_r\| ds\\
\leq&\sup_{t\in[\delta,T]}\int_0^{t-\delta}(t-s)^j e^{-\alpha(t-s)}\|\int_s^t V_r dW_r\| ds\notag\\
&+\sup_{t\in[\delta,T]}\int_{t-\delta}^t(t-s)^j e^{-\alpha(t-s)}\|\int_s^t V_r dW_r\| ds\notag\\ 
\leq&\frac{2}{\alpha^{j+1}}\left(\int_{0}^{\infty}u^{2j} e^{-u}du\right)^{1/2}e^{-\alpha \delta/2}  \sup_{\tau\in[0,T]}\|\int_0^\tau V_r dW_r\|\notag\\
&+\sup_{t\in[\delta,T]}\int_{t-\delta}^t(t-s)^j e^{-\alpha(t-s)}\|\int_s^t V_r dW_r\| ds.\notag
\end{align}
Let  $N=\min\{\ell\in\mathbb{Z}:\ell\delta\geq T\}$.  Then $P$-a.s.
\begin{align}\label{second_term_bound2}
&\sup_{t\in[\delta,T]}\int_{t-\delta}^t(t-s)^j e^{-\alpha(t-s)}\|\int_s^t V_r dW_r\| ds\\
\leq& \max_{\ell=1,...,N-1}\sup_{t\in [\ell\delta, \min\{(\ell+1)\delta,T\}]}\int_{(\ell-1)\delta}^{t}(t-s)^je^{-\alpha(t-s)}\|\int_s^tV_rdW_r\|ds\notag\\
\leq&  \frac{2}{\alpha^{j+1}}\int_{0}^{\infty}u^je^{-u}du\max_{\ell=1,...,N-1}\sup_{\tau\in[(\ell-1)\delta,\min\{(\ell+1)\delta,T\}]}\|\int_{(\ell-1)\delta}^{\tau}V_rdW_r\|.\notag
\end{align}

Combining \req{first_term_bound}, \req{second_term_bound1}, and \req{second_term_bound2} gives the $P$-a.s. bound
\begin{align}
\sup_{t\in[0,T]}&\int_0^t (t-s)^j e^{-\alpha(t-s)}\|\int_s^t V_r dW_r\| ds\\
\leq \frac{C_j}{\alpha^{j+1}}\bigg(&\max_{\ell=1,...,N-1}\sup_{\tau\in[(\ell-1)\delta,\min\{(\ell+1)\delta,T\}]}\|\int_{(\ell-1)\delta}^{\tau}V_rdW_r\|\notag\\
&+e^{-\alpha \delta/2}  \sup_{\tau\in[0,T]}\|\int_0^\tau V_r dW_r\|\bigg).\notag
\end{align}
The case $\delta\geq T$ is covered by bounding $\max_{\ell=1,...,N-1}$ by $\max_{\ell=1,...,N}$.
\end{proof}

\section{Frequently Used Inequalities}\label{app:ineq}
For the convenience of the reader, here we collect several inequalities that are repeatedly used in our proofs. In proofs, we will refer to them  via the abbreviations given in parentheses.

H{\"o}lder's Inequality (see, for example, Theorem 6.2 in \cite{folland2013real}):
\begin{lemma}[H]\label{lemma:Holder}
Let $(X,\mathcal{M},\mu)$ be a measure space, $1<p,q<\infty$ with $1/p+1/q=1$, and $f,g$ be measurable functions on $X$.  Then
\begin{align}
\int |fg|d\mu\leq \left(\int |f|^pd\mu\right)^{1/p}\left(\int |g|^qd\mu\right)^{1/q}.
\end{align}
\end{lemma}
When applied to counting measure on $\{1,...,N\}$, with $g_i=1$, H{\"o}lder's Inequality gives the following useful bound on finite sums (one can also obtain it from Jensen's inequality):
\begin{lemma}[HFS]\label{lemma:Holder_finite_sums}
Let $1\leq p<\infty$ and $f_i\geq 0$, $i=1,...,N$. Then
\begin{align}
\left(\sum_{i=1}^N f_i\right)^p\leq  N^{p-1} \sum_{i=1}^N f_i^p.
\end{align}
\end{lemma}
Minkowski's Inequality for Integrals (see Theorem 6.19 in \cite{folland2013real}):
\begin{lemma}[MI]\label{lemma:Minkowski_integral}
Let $(X,\mathcal{M},\mu)$ and $(Y,\mathcal{N},\nu)$ be sigma-finite measure spaces, $1\leq p<\infty$, and $f$ be a product-measurable function on $X\times Y$ that satisfies one of the following two conditions:
\begin{enumerate}
\item  $f\geq 0$,
\item$f(\cdot,y)\in L^p(\mu)$ for $\nu$-a.e. $y$ and  $y\to\|f(\cdot,y)\|_{L^p(\mu)}$ is in $L^1(\nu)$.
\end{enumerate}
Then
\begin{align}
\left(\left|\int f(x,y)\nu(dy)\right|^p\mu(dx)\right)^{1/p}\leq \int \left(\int |f(x,y)|^p\mu(dx)\right)^{1/p} \nu(dy).
\end{align}
\end{lemma}
$L^p$-Triangle Inequality (also known as Minkowski's inequality, see  Theorem 6.5 in \cite{folland2013real}):
\begin{lemma}[T]\label{lemma:triangle}
Let $(X,\mathcal{M},\mu)$ be a measure space, $1\leq p<\infty$, and $f,g$ be measurable functions on $X$.  Then
\begin{align}
\left(\int |f+g|^pd\mu\right)^{1/p}\leq \left(\int |f|^pd\mu\right)^{1/p}+\left(\int |g|^pd\mu\right)^{1/p}.
\end{align}
\end{lemma}

 Burkholder-Davis-Gundy Inequality  (see Theorem 3.28 in \cite{karatzas2014brownian}):
\begin{lemma}[BDG]\label{lemma:BDG}
  For every $p>0$ there exists constants $k_p,K_p\in(0,\infty)$ such that for all  $\mathbb{R}$-valued continuous local martingales, $M$, and all stopping times, $T$, we have
\begin{align}
 E\left[\sup_{0\leq s\leq T}|M_s|^{p}\right]\leq K_pE[\langle M\rangle_T^{p/2}],
\end{align}
where $\langle M\rangle$ denotes the quadratic variation of $M$.
\end{lemma}
\noindent Recall that the quadratic variation of an It{\^o} integral of a $\mathbb{R}^k$-valued continuous, adapted process, $a_t$, with respect to an $\mathbb{R}^k$-valued Wiener process, $W_t$,  is  (using summation convention) given by
\begin{align}
\left\langle \int_0^{(\cdot)} a_j(s)  dW^j_s\right\rangle_{\!T} =\int_0^T \|a(s)\|^2ds
\end{align}
 ($\|\cdot\|$  denotes the $\ell^2$ norm). If $M$ is $\mathbb{R}^n$-valued then one can still use Lemma \ref{lemma:BDG} to bound $E\left[\sup_{0\leq s\leq T}\|M_s\|^{p}\right]$ by first using 
\begin{align}
\sup_{0\leq s\leq T}\|M_s\|^{p}\leq D_{p,n}\sum_{j=1}^n \sup_{0\leq s\leq T}|M_s^j|^p,
\end{align}
where $D_{p,n}$ is a constant, depending only on $p$ and $n$.

\section{SDEs with Semimartingale Forcing}\label{app:gen_SDE}
Let $W_t$ be an $\mathbb{R}^k$-valued Wiener process on $(\Omega,\mathcal{F},P,\mathcal{F}_t)$, a filtered probability space satisfying the usual conditions  \cite{karatzas2014brownian}.  
In this section, we give some of the background theory of SDEs of the form
\begin{align}\label{generalized_SDE}
X_t=N_t+\int_{0}^tb(s,X_s)ds+\int_0^t\sigma(s,X_s)dW_s,
\end{align}
i.e., SDEs where the initial condition is generalized to a time-dependent, continuous semimartingale forcing term, $N_t$.  Much of the following can be found in \cite{protter2013stochastic}, with the generalization to SDEs with explosions adapted from \cite{hsu2002stochastic}. Both of these references discuss the generalization where $W_t$ is replaced by a more general driving semimartingale, but we do not need that extension here.

The main existence and uniqueness result for \req{generalized_SDE} mirrors that of the more standard SDE theory:
\begin{theorem}\label{gen_sde_existence}
Let $U\subset \mathbb{R}^n$ be open and $\sigma:[0,\infty)\times U\rightarrow \mathbb{R}^{n\times k}$, $b:[0,\infty)\times U\rightarrow\mathbb{R}^n$  satisfy the following:
\begin{enumerate}
\item $b,\sigma$ are measurable.
\item For every $T>0$ and compact $C\subset U$ there exists $K_{T,C}>0$ such that for all $t\in[0,T]$, $x,y\in C$ we have
\begin{align}\label{local_bounded}
\sup_{t\in[0,T],x\in C}\|b(t,x)\|+\sup_{t\in[0,T],x\in C}\|\sigma(t,x)\|\leq K_{T,C}.
\end{align}
\item For every $T>0$ and compact $C\subset U$ there exists $L_{T,C}>0$ such that for all $t\in[0,T]$, $x,y\in C$ we have
\begin{align}\label{local_lip}
\|b(t,x)-b(t,y)\|+\|\sigma(t,x)-\sigma(t,y)\|\leq L_{T,C}\|x-y\|,
\end{align}
i.e., $b(t,x)$ and $\sigma(t,x)$ are locally Lipschitz in $x$, uniformly in $t$ on compact intervals.
\end{enumerate}

Then for any continuous semimartingale $N_t$ with $N_{0}$ valued in $U$, the SDE
\begin{align}
X_t=N_{t}+\int_{0}^tb(s,X_s)ds+\int_{0}^t\sigma(s,X_s)dW_s
\end{align}
has a unique (pathwise) maximal solution up to a stopping time, $e$, called the explosion time. For every $\omega\in\Omega$, $e$ satisfies one of the following:
\begin{enumerate}
\item $e(\omega)=\infty$,
\item There exists a subsequence $t_n\nearrow e(\omega)$ with $\lim_{n\rightarrow\infty} X_{t_n}(\omega)=\infty$,
\item There exists a subsequence $t_n\nearrow e(\omega)$ with $\lim_{n\rightarrow\infty} d(X_{t_n}(\omega),\partial U)=0$.
\end{enumerate}

\end{theorem}

As with standard SDEs, non-explosion of solutions follows when the drift and diffusion are linearly bounded:
\begin{corollary}\label{lin_growth_corollary}
Let $\sigma:[0,\infty)\times \mathbb{R}^n\rightarrow \mathbb{R}^{n\times k}$, $b:[0,\infty)\times \mathbb{R}^n\rightarrow\mathbb{R}^n$  be continuous and satisfy the local Lipschitz property \req{local_lip}. Suppose we also have the following linear growth bound:\\
For each $T>0$ there exists $L_T>0$ such that
\begin{align}
\sup_{t\in[0,T]}(\|b(t,x)\|+\|\sigma(t,x)\|)\leq L_T(1+\|x\|).
\end{align}

Then for any continuous semimartingale, $N_t$,  the SDE
\begin{align}
X_t=N_{t}+\int_{0}^tb(s,X_s)ds+\int_{0}^t\sigma(s,X_s)dW_s
\end{align}
has a unique maximal solution and it is defined for all $t\geq 0$, i.e., its explosion time is $e=\infty$ a.s.

\end{corollary}

We will also need a generalization of the theory of Lyapunov functions to the current setting; it is needed to prove non-explosion for the hierarchy of approximating equations when the assumption of bounded forcing is relaxed.
\begin{theorem}\label{gen_liap_thm}
Let $U\subset \mathbb{R}^n$ be open, $W_t$ be an $\mathbb{R}^k$-valued Wiener process.  Suppose $b:[0,\infty)\times U\rightarrow \mathbb{R}^n$ and $\sigma:[0,\infty)\times U\rightarrow\mathbb{R}^{n\times k}$ are continuous and satisfy the local Lipschitz property \req{local_lip}.

Let $X_{0}$ be an $\mathcal{F}_{0}$-measurable random variable valued in $U$, $a:[0,\infty)\times\Omega\to\mathbb{R}^n$ and $c:[0,\infty)\times\Omega\to\mathbb{R}^{n\times k}$ be pathwise continuous, adapted processes, and let $N_t$ be the continuous semimartingale
\begin{align}
N_t=X_{0}+\int_{0}^t a_sds+\int_{0}^tc_s dW_s.
\end{align}

Suppose we have a $C^{1,2}$ function $V:[0,\infty)\times U\rightarrow [0,\infty)$ and measurable functions $C,M:[0,\infty)\to[0,\infty)$ that satisfy:
\begin{enumerate}
\item $M(t)$ and $C(t)$ are integrable on compact subsets of $[0,\infty)$.
\item For any $t$ and any $R>0$ there exists a compact $C\subset U$ and $\delta>0$ such that $V(s,x)\geq R$ for all $(s,x)\in [t-\delta,t]\times C^c$.
\item
\begin{align}
L[V](t,x)\equiv& \partial_tV(t,x)+b^i(t,x)\partial_{x^i}V(t,x)+\frac{1}{2}\Sigma^{ij}(t,x)\partial_{x_i}\partial_{x^j}V(t,x)\notag\\
\leq& M(t)+C(t)V(t,x),\notag
\end{align}
where $\Sigma^{ij}=\sum_\rho\sigma^i_\rho\sigma^j_\rho$,
\item $\|\nabla_x V(t,x)\|\leq M(t)+C(t)V(t,x)$,
\item $\|D^2_x V(t,x)\|(1+\|\sigma(t,x)\|)\leq M(t)+C(t)V(t,x)$.
\end{enumerate}

Then the unique maximal solution to the SDE
\begin{align}\label{liap_thm_sde}
X_t=N_{t}+\int_{0}^tb(s,X_s)ds+\int_{0}^t\sigma(s,X_s)dW_s
\end{align}
has explosion time $e=\infty$ a.s.  We call $V$ a Lyapunov function for the SDE \req{liap_thm_sde}.
\end{theorem}
\begin{proof}
Existence of a solution, $X_t$, up to explosion time, $e$, follows from Theorem \ref{gen_sde_existence}.   Let $U_n$ be precompact open sets with $\overline{U_n}\subset U_{n+1}\subset U$ and $\cup_n U_n=U$.  By looking at the equation on the events $\{X_{0}\in U_{n}\setminus U_{n-1}\}$ it suffices to suppose $X_{0}$ is contained in a compact subset of $U$ (say, $U_1$).

Define $\eta_m=\inf\{t:\|a_t\|\geq m\}\wedge\inf\{t:c_t\geq m\}$. $a_t$ and $c_t$ are continuous and adapted, so $\eta_m$ are stopping times.   Since $\eta_m$ increase to infinity, proving that there is no explosion with $N_t$ replaced by $N^m_t\equiv N_t^{\eta_m}$ for each $m$ will imply that $e=\infty$.

Therefore we can fix $m$ and consider $X$, the solution to 
\begin{align}
X_t=N^{m}_{t}+\int_{0}^tb(s,X_s)ds+\int_{0}^t\sigma(s,X_s)dW_s,
\end{align}
 with explosion time $e$.
 
 Define the stopping times $\tau_n=\inf\{t:X_t\in U_n^c\}\wedge n$ and note that $\tau_n<e$ a.s and $\|X^{\tau_n}_t\|\leq \sup_{x\in \overline{U_n}}\|x\|$. The continuous semimartingales $X^{\tau_n}$ are solutions to
\begin{align}
X^{\tau_n}_t=&X_{0}+\int_{0}^{t\wedge\tau_n} 1_{s\leq \eta_m}a_sds+\int_{0}^{t\wedge\tau_n}1_{s\leq \eta_m}c_s dW_s\\
&+\int_{0}^{t\wedge\tau_n}  b(s,X^{\tau_n}_s)ds+\int_{0}^{t\wedge\tau_n} \sigma(s,X^{\tau_n}_s)dW_s,\notag
\end{align}
hence It{\^o}'s Lemma implies
\begin{align}
&V(t\wedge\tau_n,X_t^{\tau_n})-V(0,X_{0})\\
=&\int_{0}^{t\wedge\tau_n}\partial_sV(s,X_s^{\tau_n})ds+\int_{0}^{t\wedge\tau_n}\partial_{x^i}V(s,X^{\tau_n}_s)d(X^{\tau_n})^i_s\notag\\
&+\frac{1}{2}\int_{0}^{t\wedge\tau_n}\partial_{x^i}\partial_{x^j}V(s,X^{\tau_n}_s)d[(X^{\tau_n})^i,(X^{\tau_n})^j]_s\notag
\end{align}
\begin{align}
\leq &\int_{0}^{t\wedge\tau_n} M(s)ds+\int_{0}^{t\wedge\tau_n}C(s)V(s,X_s^{\tau_n})ds\notag\\
&+\int_{0}^{t\wedge\tau_n} \partial_{x^i}V(s,X^{\tau_n}_s)(1_{s\leq \eta_m}(c_s)^i_j+\sigma^i_j(s,X^{\tau_n}_s))dB^j_s\notag\\
&+\int_{0}^{t\wedge \tau_n} 1_{s\leq \eta_m}a_s^i\partial_{x^i}V(s,X^{\tau_n}_s)ds\notag\\
&+\frac{1}{2}\int_{0}^{t\wedge\tau_n}1_{s\leq\eta_m}\partial_{x^i}\partial_{x^j}V(s,X^{\tau_n}_s)((\sigma c)^{ij}(s,X_s^{\tau_n})+(c\sigma)^{ij}(s,X_s^{\tau_n})+C^{ij}_s)ds.\notag
\end{align}
Note that if $\eta_m>0$ then $1_{s\leq \eta_m}\|a_s\|\leq m$, $1_{s\leq \eta_m}\|c_s\|\leq m$ and if $\eta_m=0$ then the integrals involving $1_{s\leq\eta_m}$ are zero.  Therefore
\begin{align}
&V(t\wedge\tau_n,X_t^{\tau_n})-V(0,X_{0})\\
\leq &\int_{0}^{t\wedge\tau_n} M(s)ds+\int_{0}^{t\wedge\tau_n}C(s)V(s,X_s^{\tau_n})ds\notag\\
&+\int_{0}^{t\wedge\tau_n} \partial_{x^i}V(s,X^{\tau_n}_s)(1_{s\leq \eta_m}(c_s)^i_j+\sigma^i_j(s,X^{\tau_n}_s))dB^j_s\notag\\
&+\frac{1}{2}\int_{0}^{t\wedge\tau_n}\|D^2V(s,X^{\tau_n}_s)\|(2m\|\sigma(s,X_s^{\tau_n})\|+m^2)ds\notag\\
&+\int_{0}^{t\wedge \tau_n} m \|\nabla V(s,X^{\tau_n}_s)\|ds\notag\\
\leq&\int_{0}^{t\wedge\tau_n} M(s)ds+\int_{0}^{t\wedge\tau_n}C(s)V(s,X_s^{\tau_n})ds\notag\\
&+\int_{0}^{t\wedge\tau_n} \partial_{x^i}V(s,X^{\tau_n}_s)(1_{s\leq \eta_m}(c_s)^i_j+\sigma^i_j(s,X^{\tau_n}_s))dW^j_s,\notag
\end{align}
where we have absorbed constants into $M(s)$ and $C(s)$.

$X^{\tau_n}_s$ is valued in $U_n$, a precompact subset of $U$.  Therefore continuity of  $V$ and $\partial_{x^i}V$ imply all of these terms have finite expectations. Also
\begin{align}
&E[\int_{0}^t  |1_{s\leq t\wedge\tau_n}\partial_{x^i}V(s,X^{\tau_n}_s)(1_{s\leq \eta_m} (c_s)^i_j+\sigma^i_j(s,X^{\tau_n}_s)|^2ds]<\infty
\end{align}
for all $t$, implying the stochastic integral is a martingale.   Therefore
\begin{align}
&E[V(t\wedge\tau_n,X_t^{\tau_n})]\\
\leq& E[V(0,X_{0})]+\int_{0}^tM(s)ds+\int_{0}^t C(s)E[V(s\wedge\tau_n,X^{\tau_n}_s)]ds.\notag
\end{align}
The integrands are $L^1$, hence Gronwall's inequality implies
\begin{align}
E[V(t\wedge \tau_n,X_t^{\tau_n})]\leq \left( E[V(0,X_{0})]+\int_{0}^tM(s)ds\right)\exp\left(\int_{0}^t C(s)ds\right)
\end{align}
for all $t\geq 0$.

Taking $n\geq t$ and using Fatou's lemma gives
\begin{align}\label{gronwall_bound}
\left( E[V(0,X_{0})]+\int_{0}^tM(s)ds\right)\exp\left(\int_{0}^t C(s)ds\right)\geq  E[\liminf_{n\to\infty}V(\tau_n,X_{\tau_n}) 1_{e<t}].
\end{align}

Now take $\omega\in\Omega$ with $e(\omega)<t$. Given $R>0$ we have a compact $C\subset U$ and  a $\delta>0$ such that $V\geq R$ on $[e(\omega)-\delta,e(\omega)]\times C^c$.  Noting that  $\tau_n(\omega)\nearrow e(\omega)$ we can take $N$ large enough that for $n\geq N$ we have $\tau_n(\omega)\in [e(\omega)-\delta,e(\omega)]$ and $C\subset U_n$.  Therefore $V(\tau_n(\omega),X_{\tau_n}(\omega))\geq R$ for $n\geq N$. So $\liminf_{n\to\infty}V(\tau_n(\omega),X_{\tau_n}(\omega))\geq R$, i.e., $\liminf_{n\to\infty}V(\tau_n,X_{\tau_n}) 1_{e<t}=\infty 1_{e<t}$.  But we have a finite upper bound \req{gronwall_bound} so we must have $P(e<t)=0$. $t\geq 0$ was arbitrary and so $e=\infty$ a.s. 
\end{proof}

 \section{Proof of Theorem \ref{thm:conv_in_prob}}\label{app:unbounded_F_proof}
In this section we provide a proof of Theorem \ref{thm:conv_in_prob}, which extends Theorem \ref{convergence_theorem} to unbounded forces, at the cost of weakening the convergence mode to convergence in probability.  Recall that here, we are working under Assumption \ref{unbounded_assump}.  First, we require several lemmas:

Assumption \ref{unbounded_assump} is sufficient to ensure non-explosion of solutions to the Langevin equation.  This can be shown by constructing Lyapunov functions:
\begin{lemma}\label{lemma:global_exist1}
 Given Assumption \ref{unbounded_assump}, there exist unique global in time solutions $(q_t^m,u_t^m)$ to \req{q_eq}-\req{u_eq} and $q_t$ to \req{q_SDE}.
 \end{lemma}
\begin{proof}
Despite the slightly different assumptions made here, the proof in Appendix C of  \cite{BirrellHomogenization} goes through essentially unchanged.  We omit the details.
\end{proof}

For $y$ a continuous semimartingale, we define $z^m_t[y]$ and $R^{m}_t[y]$ as in Definition \ref{main_def}.  The following two  properties will be needed:
\begin{lemma}\label{lemma:R_stop}
If  $\eta$ is a stopping time and  $y,\tilde y$ are continuous semimartingales that satisfy $y^\eta_t=\tilde y^\eta_t$ then
\begin{align}
R_{t\wedge\eta}^m[y]=R_{t\wedge\eta}^{m}[\tilde y]
\end{align}
for all $t\geq 0$, $P$-a.s.

\end{lemma}
\begin{proof}
The proof is a straightforward use of the formulas in Definition \ref{main_def}.
\end{proof}
\begin{lemma}\label{tilde_Y_lemma}
 Define $\tilde Y$ to be the set of continuous semimartingales of the form
\begin{align}\label{semi_mart_form}
y_t=y_0+\int_0^t a_s ds+\int_0^t c_sdW_s
\end{align}
where $y_0$ is $\mathcal{F}_0$-measurable and $a\!:\![0,\infty)\times \Omega\to\mathbb{R}^n$, $c\!:\![0,\infty)\times \Omega\to\mathbb{R}^{n\times k}$ are pathwise continuous, adapted processes. 

If $y\in\tilde Y$ then  $z^m_t[y]\in\tilde Y$ and $R^{m}_t[y]\in\tilde Y$.
\end{lemma}
\begin{proof}
The set of semimartingales of the form \req{semi_mart_form} is a vector space and, using integration by parts, one can see that is closed under  multiplication by $\mathbb{R}$-valued processes of the form \req{semi_mart_form} (i.e., with $n=1$), and contains  $z_t^m[y]$  for any continuous semimartingale $y$.

The result then follows for $R^m_t[y]$ by noting that Assumption \ref{unbounded_assump} implies all of the integrands are pathwise continuous, adapted, and that $\tilde \gamma^{-1}(t,q)$, $Q^{ikl}(t,q)$, and $G^{a,b}_{kl}(t,q)$ are $C^2$.  The latter allows It{\^ o}'s Lemma to be applied to $\tilde \gamma^{-1}(t,y_t)$ etc., yielding terms in $\tilde Y$, provided that $y\in\tilde Y$.
\end{proof}

We also need to know that solutions to the SDE defining the hierarchy exist under the current weakened assumptions:
\begin{lemma}\label{lemma:global_exist2}
Under Assumption \ref{unbounded_assump}, for any $y\in \tilde Y$ (defined in Lemma \ref{tilde_Y_lemma}) there is a unique continuous semimartingale, $x_t$, defined for all $t\geq 0$ that solves
\begin{align}\label{relaxed_hierarchy_eq}
x_t=&q_0+\int_0^t\tilde \gamma^{-1}(s,x_s)F(s,x_s)ds+\int_0^tS(s,x_s)ds\\
&+\int_0^t\tilde \gamma^{-1}(s,x_s)\sigma(s,x_s) dW_s +\sqrt{m}R^{m}_t[y].\notag
\end{align}
We also have $x\in\tilde Y$.

\end{lemma}
\begin{proof}
 $\tilde\gamma^{-1}F+S$ and $\tilde\gamma^{-1}\sigma$ are continuous and satisfy the local Lipschitz property, \req{local_lip}. Lemma \ref{tilde_Y_lemma} implies $R^m[y]$ is a continuous semimartingale (in fact, $R^m[y]\in\tilde Y$).  Therefore Theorem \ref{gen_sde_existence} shows  a  unique maximal solution exists up to explosion time. 

 One can check that the function
 \begin{align}
 (t,q)\to \|q\|^2+\tilde V(t,q),
 \end{align}
 where $\tilde V$ was defined in \req{tilde_V_def},  satisfies all the conditions required by Theorem \req{gen_liap_thm} to make it a Lyapunov function for the SDE \req{relaxed_hierarchy_eq}, thereby proving $x_t$ has explosion time $e=\infty$.   $R^m[y]\in \tilde Y$  together with \req{relaxed_hierarchy_eq} shows that $x\in\tilde Y$ as well.
 \end{proof}

We are now ready to prove Theorem \ref{thm:conv_in_prob}:
\begin{proof}
By Lemma \ref{lemma:global_exist1}, there exist unique global in time solutions $(q_t^m,u_t^m)$ to \req{q_eq}-\req{u_eq} and $q_t\in \tilde Y$ ($\tilde Y$ was defined in Lemma \ref{tilde_Y_lemma}) to \req{q_SDE}, and by induction, Lemma \ref{lemma:global_exist2} gives globally defined continuous semimartingale solutions to the approximation hierarchy, \req{relaxed_hierarchy_eq}.

Let $\chi:\mathbb{R}^n\to [0,1]$ be a $C^\infty$ bump function, equal to $1$ on  $\overline{B_1(0)}\equiv\{\|q\|\leq 1\}$ and zero outside $\overline{B_2(0)}$. Given $r>0$ let $\chi_r(q)=\chi(q/r)$.  Define
\begin{align}
&V_r(t,q)=\chi_r(q)V(t,q), \hspace{2mm} \tilde F_r(t,q)=\chi_r(q) \tilde F(t,q), \hspace{2mm}  \psi_r(t,q)=\chi_r(q)\psi(t,q),\notag\\
&\gamma_r(t,q)=\chi_r(q)\gamma(t,q)+(1-\chi_r(q))\lambda I.
\end{align}
For each $r>0$, replacing $V$ with $V_r$, $F$ with $F_r$ etc., we arrive at an SDE satisfying the hypotheses of Theorem \ref{convergence_theorem}.   We will call these the cutoff systems.

Let $R^{r,m}_t[y]$ denote \req{R_main_def}, with $V$ replaced by $V_r$, etc. All of these objects and their derivatives agree on $\overline{B_r(0)}$, so for any continuous semimartingale, $y$, if we let $\eta_r^y=\inf\{t:\|y_t\|\geq r\}$, we have 
\begin{align}\label{R_cutoff_equality}
R_{t\wedge\eta_r^y}^m[y]=R_{t\wedge\eta_r^y}^{r,m}[y]
\end{align}
 for all $t\geq 0$, $P$-a.s. 

  Let $(q_t^{r,m},u_t^{r,m})$ be the solutions to the cutoff system, $q_t^r$ the solution to the corresponding homogenized equation, and $q_t^{r,\ell,m}$ the solutions to the corresponding approximating hierarchy, all using the same initial conditions as the system without the cutoff.

For each $r>R$ define the stopping times 
\begin{align}
&\eta^m_r=\inf\{t:\|q_t^m\|\geq r\},\,\, \eta^{\ell,m}_r=\inf\{t:\|q^{\ell,m}_t\|\geq r\},\\
& \eta^{r,\ell,m}_r=\inf\{t:\|q^{r,\ell,m}_t\|\geq r\},\notag
\end{align}
 and
 \begin{align}
 \tau^{\ell,m}_r=\eta^{\ell,m}_r\wedge\eta^{\ell-1,m}_r\wedge...\wedge \eta^{1,m}_r,\,\, \tau^{r,\ell,m}_r=\eta^{r,\ell,m}_r\wedge\eta^{r,\ell-1,m}_r\wedge...\wedge \eta^{r,1,m}_r.
 \end{align}
 Note that $\eta^{1,m}_r=\inf\{t:\|q_t\|\geq r\}\equiv \eta_r$ is independent of $m$.  Finally, define $\sigma^{\ell,m}_r=\tau^{\ell,m}_r\wedge\eta^m_r$, the first exit time for any of the position processes up to level $\ell$ of the hierarchy.

 The drifts and diffusions of the modified and unmodified SDEs agree on the ball $\{\|q\|\leq r\}$.  Therefore, using induction on $\ell$, Lemma \ref{lemma:R_stop}, \req{R_cutoff_equality}, and pathwise uniqueness of solutions,  we see  that the driving semimartingales of the hierarchy up to $\ell$ for both the original and cutoff systems agree up to the stopping time $\tau^{\ell,m}_r$ and 
\begin{align}\label{q_eps_unique}
q^m_{t\wedge\eta^m_r}=q^{r,m}_{t\wedge\eta^m_r} \text{ for all $t\geq 0$ a.s., }
\end{align}
\begin{align}\label{q_eps_unique2}
\tau^{\ell,m}_r=\tau^{r,\ell,m}_r a.s., \,\,\text{ and }\,\, q^{\ell,m}_{t\wedge\tau^{\ell,m}_r}= q^{r,\ell,m}_{t\wedge\tau^{\ell,m}_ r}  \text{ for all $t\geq 0$ a.s.}
\end{align}

Fixing $r>0$ and using \req{q_eps_unique} and \req{q_eps_unique2}, for any $T>0$, $\delta>0$, $\epsilon>0$, $\ell\in\mathbb{Z}^+$ we can calculate 
\begin{align}\label{eq:conv_prob1}
&P\left(\frac{\sup_{t\in[0,T]}\|q_t^m-q^{\ell,m}_t\|}{m^{\ell/2-\epsilon}}>\delta\right)\\
=&P\left(\sigma^{\ell,m}_r> T,\frac{\sup_{t\in[0,T]}\|q_{t\wedge\eta^m_r}^{r,m}-q^{r,\ell,m}_{t\wedge\tau^{\ell,m}_r}\|}{m^{\ell/2-\epsilon}}>\delta\right)\notag\\
&+P\left(\sigma^{\ell,m}_r\leq T,\frac{\sup_{t\in[0,T]}\|q_t^m-q^{\ell,m}_t\|}{m^{\ell/2-\epsilon}}>\delta\right)\notag\\
\leq&P\left(\frac{\sup_{t\in[0,T]}\|q_{t}^{r,m}-q^{r,\ell,m}_{t}\|}{m^{\ell/2-\epsilon}}>\delta\right)+P\left(\sigma^{\ell,m}_r\leq T\right).\notag
\end{align}
The first term, involving the cutoff system, converges to zero as $m\to 0$ by  Markov's inequality and the convergence result for bounded forces, \req{hierarchy_conv_rate}.  Hence we focus on the second term.  We note that the only essential difference between the argument below and the similar computation in the proof of Theorem 6.1 from \cite{BirrellHomogenization} is the need to consider all processes in the hierarchy up to level $\ell$, and not just the 
processes, $q_t^m$ and $q_t^{\ell,m}$, that  were being compared in \req{eq:conv_prob1}.  This is due to the iterative construction of each level in the hierarchy from the solution at the previous level. The second term can be bounded as follows:
\begin{align}\label{eq:r_ball_prob}
&P\left(\sigma^{\ell,m}_r\leq T\right)\\
\leq&P(\eta_r\leq T)+P\left(\eta^{m}_{r}\leq T,\eta_r>T, \sup_{t\in[0,T]}\|q_t^r-q_t^{r,m}\|\leq 1\right)\notag\\
& +\sum_{k=2}^\ell P\left(\tau^{k-1,m}_r>T,\eta^{k,m}_r\leq T, \sup_{t\in[0,T]}\|q_t^r-q_t^{r,k,m}\|\leq 1\right)\notag\\
&+P\left(\sup_{t\in[0,T]}\|q_t^r-q_t^{r,m}\|>1\right)+\sum_{k=2}^\ell P\left(\sup_{t\in[0,T]}\|q_t^r-q_t^{r,k,m}\|> 1\right)\notag\\
\leq&P\left(\sup_{t\in[0,T]}\|q_t\|\geq r\right)+P\left(\eta^{m}_{r}\leq T, \|q_{T\wedge\eta_r^m}-q_{T\wedge\eta_r^m}^{m}\|\leq 1\right)\notag\\
& +\sum_{k=2}^\ell P\left(\eta^{k,m}_r\leq T, \|q_{T\wedge \eta^{k,m}_r}-q_{T\wedge\eta^{k,m}_r}^{k,m}\|\leq 1\right)\notag\\
&+E\left[\sup_{t\in[0,T]}\|q_t^r-q_t^{r,m}\|\right]+\sum_{k=2}^\ell E\left[\sup_{t\in[0,T]}\|q_t^r-q_t^{r,k,m}\|\right],\notag
\end{align}
where we again used the uniqueness results, \req{q_eps_unique}-\req{q_eps_unique2}. The terms in the last line  go to zero as $m\to 0$, as seen from the triangle inequality and \req{hierarchy_conv_rate}.

On the event where $\eta_r^m\leq T$ and $\|q_{T\wedge\eta_r^m}-q^m_{T\wedge\eta_r^m}\|\leq 1$ we have $\|q^m_{\eta^m_r}\|\geq r$ and
\begin{align}
\|q_{\eta_r^m}\|\geq \|q^m_{T\wedge \eta^m_r}\|-\|q_{T\wedge\eta_r^m}-q^m_{T\wedge\eta_r^m}\|\geq r-1.
\end{align}
Hence $\sup_{t\in[0,T]}\|q_t\|\geq r-1$ on this event. Similarly,
\begin{align}
\left\{\eta^{k,m}_r\leq T, \|q_{T\wedge \eta^{k,m}_r}-q_{T\wedge\eta^{k,m}_r}^{k,m}\|\leq 1\right\}\subset\left\{\sup_{t\in[0,T]}\|q_t\|\geq r-1\right\}.
\end{align}

Therefore we obtain
\begin{align}
&\limsup_{m\to 0}P\left(\frac{\sup_{t\in[0,T]}\|q_t^m-q^{\ell,m}_t\|}{m^{\ell/2-\epsilon}}>\delta\right)\\
\leq&P\left(\sup_{t\in[0,T]}\|q_t\|\geq r\right) +\ell P\left(\sup_{t\in[0,T]}\|q_t\|\geq r-1\right)\notag\\
\leq& (\ell+1)P\left(\sup_{t\in[0,T]}\|q_t\|\geq r-1\right).\notag
\end{align}
This holds for all $r>0$ and  non-explosion of $q_t$ implies that 
\begin{align}
P\left(\sup_{t\in[0,T]}\|q_t\|\geq r-1\right)\to 0
\end{align}
 as $r\to\infty$,  hence we have proven the claimed result.
 \end{proof}

\subsection*{Acknowledgments}

J.W. was partially supported by NSF grant DMS 1615045. J.B. would like to thank Giovanni Volpe for suggesting this problem. J.B and J.W. would like to warmly thank the reviewers for their careful reading and many helpful suggestions for improving the presentation of this work.
\bibliographystyle{unsrt}
\bibliography{refs}

\end{document}